\documentclass[11pt]{article}

\usepackage{fullpage,amsmath,xspace,amssymb,epsfig}
\usepackage{bbm}
\usepackage{amsthm}
\usepackage{hyperref}

\newtheorem{Thm}{Theorem}
\newtheorem{Lem}[Thm]{Lemma}
\newtheorem{Cor}[Thm]{Corollary}
\newtheorem{Prop}[Thm]{Proposition}

\newtheorem{Conj}[Thm]{Conjecture}
\newtheorem{Def}{Definition}
\newtheorem{Fact}[Thm]{Fact}

\newcommand\mbR{\mbox{$\mathbb{R}$}}

\newcommand\C{\mbox{\sf {C}}\xspace}
\newcommand\D{\mbox{\sf {D}}\xspace}

\newcommand\dcc{\mbox{$\sf {D^{CC}}$}\xspace}

\newcommand\gf{\mbox{$\mathbb{F}_2$}\xspace}

\newcommand\B{\{0,1\}}     
\newcommand\pmB{\{+1,-1\}}     
\newcommand\Bn{\{0,1\}^n}
\newcommand\BntB{\{0,1\}^n\rightarrow \{0,1\}}
\newcommand {\ie} {\textit{i.e.}\xspace}
\newcommand {\st} {\textit{s.t.}\xspace}
\newcommand {\etal} {\textit{et al.}\xspace}

\newcommand\pr{\mbox{\bf Pr}}
\newcommand\av{\mbox{\bf{\bf E}}}

\newcommand\alice{\mbox{\sf {Alice}}\xspace}
\newcommand\bob{\mbox{\sf {Bob}}\xspace}

\newcommand\rank{\mbox{\tt {rank}}\xspace}
\newcommand\gran{\mbox{\tt {gran}}\xspace}

\newcommand\supp{\mbox{\tt {supp}}\xspace}
\newcommand\range{\mbox{\tt {range}}\xspace}
\newcommand\codim{\mbox{\tt {co-dim}}\xspace}

\newcommand\osize{\mbox{$\oplus$-\tt{size}}\xspace}
\newcommand\bias{\mbox{\tt {bias}}\xspace}

\newcommand\fn[2]{\| \hat{#1} \|_#2}
\newcommand\wfn[2]{\| \widehat{#1} \|_#2}

\newcommand{\polylog}[1]{\mathrm{polylog}{#1}}

\newtheorem*{Thm-lighttail}{Theorem~\ref{thm:lighttail}}
			
\begin{document}
\title{\bf Fourier sparsity, spectral norm, and the Log-rank conjecture} 
\author{Hing Yin Tsang\thanks{The Chinese University of Hong Kong, Shatin, NT, Hong Kong. Email: {\tt hytsang@cse.cuhk.edu.hk}} 
\and Chung Hoi Wong\thanks{The Chinese University of Hong Kong, Shatin, NT, Hong Kong. Email: {\tt hoiy927@gmail.com}} 
\and Ning Xie\thanks{SCIS, Florida International University, Miami, FL 33199, USA. Email: {\tt nxie@cis.fiu.edu}} 
\and Shengyu Zhang\thanks{The Chinese University of Hong Kong, Shatin, NT, Hong Kong. Email: {\tt syzhang@cse.cuhk.edu.hk}}}
\date{}
\setcounter{page}{0}
\maketitle

\begin{abstract}
We study Boolean functions with sparse Fourier coefficients or small spectral norm,
and their applications to the Log-rank Conjecture for XOR functions $f(x\oplus y)$ --- 
a fairly large class of functions including well studied ones such as Equality and Hamming Distance. 
The rank of the communication matrix $M_f$ for such functions is exactly the Fourier sparsity of $f$.
Let $d = \deg_2(f)$ be the \gf-degree of $f$ and
$\dcc(f\circ \oplus)$ stand for the deterministic communication complexity for $f(x\oplus y)$.
We show that 
\begin{enumerate}
\item $\dcc(f\circ \oplus) = O(2^{d^2/2} \log^{d-2}\fn{f}{1})$. 
In particular, the Log-rank conjecture holds for XOR functions with constant \gf-degree. 
\item $\dcc(f\circ \oplus) = O(d\fn{f}{1}) = O(\sqrt{\rank(M_f)}\log\rank(M_f))$.
This improves the (trivial) linear bound by nearly a quadratic factor.
\end{enumerate}
We obtain our results through a degree-reduction protocol based on a variant of polynomial rank, 
and actually conjecture that the communication cost of our protocol is at most $\log^{O(1)}\rank(M_f)$. 
The above bounds are obtained from different analysis for the number of parity queries 
required to reduce $f$'s \gf-degree.
Our bounds also hold for the parity decision tree complexity of $f$, 
a measure that is no less than the communication complexity. 
	
Along the way we also prove several structural results about Boolean functions with small Fourier sparsity $\fn{f}{0}$ 
or spectral norm $\fn{f}{1}$, which could be of independent interest. 
For functions $f$ with constant \gf-degree, we show that: 
1) $f$ can be written as the summation of quasi-polynomially many indicator functions of subspaces with $\pm$-signs, 
improving the previous doubly exponential upper bound by Green and Sanders;
2) being sparse in Fourier domain 
is polynomially equivalent to having a small parity decision tree complexity; and 
3) $f$ depends only on 
$\polylog{\fn{f}{1}}$ linear functions of input variables.
For functions $f$ with small spectral norm, we show that: 
1) there is an affine subspace of co-dimension $O(\fn{f}{1})$ 
on which $f(x)$ is a constant, and 
2) there is a parity decision tree of depth $O(\fn{f}{1}\log\fn{f}{0})$.
\end{abstract}
\newpage

\section{Introduction}
\paragraph{Fourier analysis of Boolean functions.} 
Fourier analysis has been widely used in theoretical computer science to study Boolean functions
with applications in PCP, property testing, learning, circuit complexity, coding theory, 
social choice theory and many more; see~\cite{O12} for a comprehensive survey.
The Fourier coefficients of a Boolean function measure the function's correlations with parity functions;
the distribution as well as various norms of Fourier spectrum
have been found to be related to many complexity measures of the function.
However, another natural measure,
Fourier sparsity -- \ie the number of non-zero Fourier coefficients -- has been much less studied.
It seems to be of fundamental interest to understand properties of functions that are 
Boolean in the function domain and, at the same time, sparse in the Fourier domain. 
In particular, what Boolean functions have sparse Fourier spectra? 
Being sparse in the Fourier domain should imply that the function is simple, but in which aspects?
Gopalan et al.~\cite{GOS+11} studied the problem of testing Fourier sparsity and low-dimensionality
and revealed several interesting structural results for Boolean functions having or close to
having sparse Fourier spectra.
In a related setting, Green and Sanders~\cite{GS08} showed that Boolean functions with a small
spectral norm (\ie the $\ell_1$-norm of the Fourier spectrum) can be decomposed into
a small number of signed indicator functions of subspaces. 
However, the number of subspaces in their bound is doubly exponential in terms of the
function's spectral norm, thus makes their result hard to apply in many computer science related problems.

\paragraph{The Log-rank Conjecture in communication complexity.}
In a different vein, Fourier sparsity also naturally arises in the 
study of Log-rank Conjecture in communication complexity. 
Communication complexity quantifies the minimum amount of communication needed for computation on inputs 
distributed to different parties \cite{Yao79,KN97}. 
In a standard scenario, two parties \alice and \bob each hold an input $x$ and $y$, respectively, 
and they desire to compute a function $f$ on input $(x,y)$ by as little communication as possible. 
Apart from its own interest as a question about distributed computation, 
communication complexity has also found numerous applications in proving lower bounds in complexity theory, 
as well as connections to linear algebra, graph theory, etc. 

Of particular interest are lower bounds of communication complexity, and one of the most widely used methods is 
based on the rank of the communication matrix $M_f = [f(x,y)]_{x,y}$; 
see \cite{LS09} for an extensive survey on classical and quantum lower bounds proved by rank and its variations 
(such as the approximate rank and its equivalence $\gamma_2$-norm). 
Since it was shown 30 years ago \cite{MS82} that $\log\rank(M_f)$ is a lower bound of the deterministic 
communication complexity $\dcc(f)$, the tightness of the lower bound has long been an important open question. 
The Log-rank Conjecture, proposed by Lov{\'a}sz and Saks \cite{LS88}, 
asserts that the lower bound is polynomially tight for all total Boolean functions $f$ -- 
namely $\dcc(f) \leq \log^c\rank(M_f)$ for some absolute constant $c$. 
As one of the most important problems in communication complexity, 
the conjecture links communication complexity -- a combinatorially defined quantity, 
to matrix rank -- a much better understood measure in linear algebra. 
Should the conjecture hold, understanding the communication complexity is more or less 
reduced to a usually much easier task of calculating matrix ranks. 
The conjecture is also known to be equivalent to many other conjectures~\cite{LS88,Lov90,Val04,ASTS+03}.

Despite its importance, Log-rank Conjecture is also notoriously hard to attack. 
Nisan and Wigderson~\cite{NW95} showed that to prove the conjecture, 
it is sufficient to show a seemingly weaker statement about 
the existence of a large monochromatic rectangle. 
In the same paper, they also exhibited an example $f$ for which 
$\log\rank(M_f) = O(\dcc(f)^\alpha)$ where $\alpha = \log_3 2 = 0.63...$, 
later improved by Kushilevitz to $\alpha = \log_6 3 = 0.61...$ (also in \cite{NW95}). 
The best upper bound for the $\dcc(f)$ in terms of rank is 
$\dcc(f) \leq (\log\frac{4}{3})\rank(M_f)$ \cite{KL96,Kot97}. 
Recently, assuming the Polynomial Freiman-Ruzsa conjecture in additive combinatorics, 
Ben-Sasson, Lovett and Ron-Zewi gave in \cite{BLR12} a better upper bound 
$\D(f) \leq O(\rank(M_f)/\log\rank(M_f))$. 

\paragraph{Communication complexity of XOR functions.}
In view of the difficulty of the Log-rank Conjecture in its full generality, 
Shi and Zhang~\cite{ZS10} initiated the study of communication complexity of 
a special class of functions called \emph{XOR functions}.

\begin{Def}\label{def:XOR}
We say $F(x,y): \{0,1\}^{n} \times \{0,1\}^{n} \to \{0,1\}$ is an XOR function if there exists 
an $f: \{0,1\}^{n} \to \{0,1\}$ such that
for all $x$ and $y$ in  $\{0,1\}^{n}$, 
$F(x,y) = f(x\oplus y)$, where $\oplus$ is the bit-wise XOR. Denote $F$ by $f\circ \oplus$.
\end{Def}

XOR functions include important examples such as Equality and Hamming Distance,
and the communication complexity of XOR functions has recently 
drawn an increasing amount of attention \cite{ZS09,ZS10,LZ10,MO10,LLZ11,SW12,LZ13}. 
In general, the additional symmetry in the communication matrix $M_F$ should make 
Log-rank Conjecture easier for XOR functions. 
In particular, a very nice feature of XOR functions is that the rank of the communication matrix $M_{F}$ 
is exactly the Fourier sparsity of $f$, the number of nonzero Fourier coefficients $f$.
\begin{Prop}[\cite{BC99}]\label{thm:sparsity}
For XOR functions $F(x,y) = f(x\oplus y)$, it holds that $\rank(M_{F})=\|\hat f\|_0$.
\end{Prop}

Therefore the Log-rank Conjecture for XOR functions is equivalent to 
the question that whether 
every Fourier sparse\footnote{Note that if the Fourier sparsity of $f$ is large, say $2^{n^{\Omega(1)}}$,
then Log-rank Conjecture is vacuously true for $f$, 
as the communication complexity of any function is at most $O(n)$.} 
function $f$ admits an efficient communication protocol to compute $f(x\oplus y)$,
or more specifically, 
whether $\dcc(f\circ \oplus) \leq \log^{O(1)}\fn{f}{0}$ holds for every Boolean function $f$?

However, the Log-rank conjecture seems still very difficult to study even for this special class of functions.
The only previously known results are that
the Log-rank Conjecture for XOR functions holds for all \emph{symmetric} functions~\cite{ZS09},
monotone functions and linear threshold functions (LTFs)~\cite{MO10}, and $AC^0$ functions; 
see Section \ref{sec:relatedwork} for more details. 
One nice approach proposed in \cite{ZS10} is to first design an efficient \emph{parity decision tree} (PDT) 
for computing $f$, and then to simulate it by a communication protocol. 
Parity decision trees allow querying the parity of any subset of input variables 
(instead of just one input variable as in usual decision trees). 
A communication protocol can exchange two bits $\ell(x)$ and $\ell(y)$ 
(here $\ell(\cdot)$ is an arbitrary linear function)
to simulate one query $\ell(x \oplus y)$ in a PDT, 
thus $\dcc(f\circ \oplus)$ is at most twice of $\D_\oplus (f)$, 
the parity decision tree complexity of $f$. 
It is therefore sufficient to show that $\D_\oplus (f)\leq \log^{O(1)}\fn{f}{0}$
for all $f$ to prove the Log-rank Conjecture for XOR functions 
Parity decision tree complexity is an interesting complexity on its own, 
with connections to learning \cite{KM93} and 
other parity complexity measures such as parity certificate complexity and 
parity block sensitivity \cite{ZS10}. 
This approach is also appealing for the purpose of understanding Boolean functions with sparse Fourier spectra.
It is not hard to see that small $\D_\oplus (f)$ implies Fourier sparsity; 
now if $\D_\oplus (f)\leq \log^{O(1)}\fn{f}{0}$ is true, 
then functions with small Fourier sparsity also have short parity decision trees. 
Thus the elusive property of being Fourier sparse is roughly equivalent to the combinatorial 
and computational property of having small PDT. 

Back to the Log-rank conjecture, though upper bounds for $\D_\oplus(f)$ translate to 
efficient protocols for $\dcc(f\circ \oplus)$, 
the task of designing efficient PDT algorithms itself does not seem to be an easy task. 
To see this, let us examine the effect of parity queries. 
Each query ``$ t \cdot x = ?$'' basically generates two subfunctions through restriction,
and its effect on the Fourier domain can be shown to be 
$\hat f_b(s) = \hat f(s) + (-1)^b \hat f(s+t)$, 
where $f_b$ is the subfunction obtained from restricting $f$ on the half space $\{x: t\cdot x = b\}$. 
Thus the process is like to fold the spectrum of $f$ along the line $t$, 
and we hope that the folding has many ``collisions'' in nonzero Fourier coefficients, 
namely many $s\in \supp(\hat f)$, with $s+t\in \supp(\hat f)$ as well. 
In general, small $\D_\oplus(f)$ implies that many Fourier coefficients\footnote{Technically, 
we mean characters with the corresponding Fourier coefficients being nonzero.} are ``well aligned'' 
with respect to a subspace $V$ with a small co-dimension, 
so that querying basis of $V^\bot$ make those Fourier coefficients collide. 
But the question is---\emph{Where is the subspace}?

Note that $\D_\oplus(f)$ is invariant under change of input basis, 
thus one tempting way to upper bound $\D_\oplus(f)$ is to first rotate input basis, 
and then (under the new basis) use the well-known fact that 
the standard decision tree complexity $\D(f)$ is at most $O(\deg(f)^4)$, 
where the $\deg(f)$ is the (Fourier) degree ($\max_{s: \hat f(s)\neq 0} |s|$) of $f$ \cite{BdW02}. 
Thus if $\deg(f) = \log^{O(1)}\fn{f}{0}$, then 
$\D_\oplus(f) \leq \D(f) \leq \log^{O(1)}\fn{f}{0}$. 
However, one should also note that this approach cannot handle all the Fourier sparse functions because, 
as shown in~\cite{ZS10}, there exists a functions $f$ such that 
$\D_\oplus(f) \leq \log_2 n + 4$ but $\D(f) \geq n/4$, 
the latter holds even under an arbitrary basis change 
(\ie $\min_L\D(Lf) \geq n/4$ where $Lf(x) = f(Lx)$).


\subsection{Our approach, ideas, and results}
\paragraph{Result 1: Main protocol and general conjecture.}
In previous studies of parity decision tree, 
one needs to upper bound the number of queries for \emph{all} possible execution paths. 
In this paper, we show that it suffices to prove the existence of \emph{one} short path! 
To put this into context, we need the concept of polynomial rank.
View a Boolean function $f:\BntB$ as a polynomial in $\gf[x_1, ..., x_n]$. 
Call the degree of this polynomial the \gf-degree, denoted as $\deg_2(f)$. 
The polynomial rank of $f$ is the minimum number $r$ \st 
$f$ can be written as 
\begin{equation}\label{eq:rank_def}
f = \ell_1f_1 + \cdots + \ell_rf_r + f_0,
\end{equation} 
where each $\ell_i$ is a linear function in $x$ and 
each $f_i$ is a function of \gf-degree at most $\deg_2(f)-1$. 
Now we will describe a simple PDT algorithm: query all $\ell_i(x)$ and get answers $a_i$, 
and we then face a new function $f' = \sum_{i=1}^r a_i f_i + f_0$'s. 
Recurse on this function. Note that from $f$ to $f'$, the \gf-degree is reduced by at least 1, 
and one can also show that the Fourier sparsity of $f'$ is also at most that of $f$. 
Finally, it is known that $d\leq \log\fn{f}{0}$. Putting these nice properties together, 
we know that as long as the polynomial rank of an arbitrary function $f$ is upper bounded by 
$\log^{O(1)}\fn{f}{0}$, so is $\D_\oplus(f)$. 


\begin{Conj}\label{conj:rank}
	For all $f:\BntB$, we have $\rank(f) = O(\log^c(\fn{f}{0}))$ for some $c=O(1)$. 
\end{Conj}
\begin{Thm}
	If Conjecture \ref{conj:rank} is true, then 
	\begin{enumerate}
		\item All Boolean functions with small $\fn{f}{0}$ have small parity decision tree complexity as well: 
		    $\D_\oplus(f) = O(\log^{c+1}(\fn{f}{0}))$.
		\item The Log-rank Conjecture is true for all XOR functions: $\dcc(f\circ \oplus) = O(2\log^{c+1}(\fn{f}{0}))$.
	\end{enumerate}
\end{Thm}

\paragraph{Result 2: low degree polynomials.}
Next we focus on upper bounding the polynomial rank, starting from small degrees. 
For degree-$2$ polynomials, the classic theorem by Dickson implies that 
$\rank(f) = O(\log \fn{f}{0})$. 
For degree-$3$ polynomials, Haramaty and Shpilka proved in \cite{HS10} that 
$\rank(f) = O(\log^2 (1/\|f\|_{U^3})) = O(\log^2 (1/\bias(f)))$. 
By a proper shift, we can make $\bias(f) \geq 1/\sqrt{\fn{f}{0}}$ and thus get 
$\rank(f) = O(\log^2 \fn{f}{0})$. 
For degree-$4$ polynomials, however, the bound in \cite{HS10} is exponentially worse, 
and there were no results for higher degrees. 
A natural question is: Can one prove the $\rank(f) = O(\log^{O(1)} \fn{f}{0})$ for degree-$4$ polynomials? 
Further, if it is too challenging to prove $\rank(f) \leq \log^{O(1)} \fn{f}{0}$ for  
general degree $d$ (which is at most $\log \fn{f}{0}$), 
can one prove it for constant-degree polynomials (even if the power $O(1)$ is a tower of 2's of height $d$)? 
In this paper, we show that this is indeed achievable. 
Actually, we can even replace the $\ell_0$-norm by $\ell_1$-norm\footnote{Strictly speaking, in view of the corner case of $\fn{f}{1} = 1$, one should replace $\log(\fn{f}{1})$ by $\log(\fn{f}{1}+1)$. But like in most previous papers, we omit the ``+1'' term for all $\fn{f}{1}$ in this paper for simplicity of notation.} of $\hat f$ in the bound, 
and the dependence on $d$ is ``only'' singly exponential.

\begin{Lem}\label{lem:rank by 1norm}
	For all Boolean functions $f$ with \gf-degree $d$, we have 
	$\rank(f) = O(2^{d^2/2} \log^{d-2} \fn{f}{1})$. 
\end{Lem}
The lemma immediately implies the following two results. 
\begin{Thm}\label{thm:constdeg}
If $f$ is a Boolean function of constant \gf-degree, 
then $\dcc(f\circ \oplus) \leq \log^{O(1)}\left(\rank(M_{f\circ \oplus})\right)$.
\end{Thm}

Recursively expanding Eq.\eqref{eq:rank_def} and applying the bound on ranks in Lemma~\ref{lem:rank by 1norm} gives that 
\begin{Cor}\label{cor:junta}
Every Boolean function $f$ of \gf-degree $d$ depends only on $O(2^{d^3/2}\log^{d^2}\fn{f}{1})$ 
linear functions of input variables.
\end{Cor}
Another corollary is the following. Green and Sanders proved that any $f:\BntB$ can be written as $f=\sum_{i=1}^{T} \pm \mathbbm{1}_{V_i}$, where $T = 2^{2^{O(\fn{f}{1}^4)}}$ and each $\mathbbm{1}_{V_i}$ is the indicator function of the subspace $V_i$. For constant degree polynomials, we can improve their doubly-exponential bound to quasi-polynomial.
\begin{Cor}\label{Cor:constdegGS}
	If $f:\BntB$ has constant \gf-degree, then $f = \sum_{i=1}^{T} \pm \mathbbm{1}_{V_i}$ where $T = 2^{\log^{O(1)}\fn{f}{1}}$ and each $\mathbbm{1}_{V_i}$ is the indicator function of the subspace $V_i$.
\end{Cor}

The proof of Lemma~\ref{lem:rank by 1norm} follows the general approach laid out in the Main protocol,
\ie, a rank-based degree-reduction process, with several additional twists.
First, to find a ``good'' affine subspace restricted 
on which $f$ becomes a lower degree polynomial,
we recursively apply the derivatives of $f$ to guide our search.
Second, even though our final goal is to reduce the degree of $f$, we actually achieve this
through reducing the spectral norm of $f$.
This is done by 
studying the effect of restriction on two \emph{non-Boolean} functions.
Last, in the induction step, we in fact need to prove a stronger statement about
a chain inequality involving rank, minimum parity $0$-certificate complexity $\C_{\oplus, \min}^0$, 
minimum parity $1$-certificate complexity $\C_{\oplus, \min}^1$ and parity decision tree complexity $\D_\oplus$. 
And the induction is used in a ``cyclic'' way: 
we upper bound $\min\{\C_{\oplus, \min}^0,\C_{\oplus, \min}^1\}$ by induction on 
$\max\{\C_{\oplus, \min}^0,\C_{\oplus, \min}^1\}$, which upper bounds $\rank$. 
This can then be used to show that $\D_\oplus$ is small, 
which in turn upper bounds $\max\{\C_{\oplus, \min}^0,\C_{\oplus, \min}^1\}$ to finish the inductive step.


\paragraph{Result 3: functions with small spectral norms.}
While Theorem~\ref{thm:constdeg} handles the low-degree case, the bound 
deteriorates exponentially with the \gf-degree.
Via a different approach, 
we are able to upper bound 
$\rank(f)$ by the $\ell_1$-norm of $\hat f$. 

\begin{Lem}\label{lem:l1norm}
For all $f:\BntB$, we have $\rank(f) \leq O(\fn{f}{1})$. 
\end{Lem}
In fact we prove a slightly stronger result that there exists an affine subspace of
codimension at most $O(\fn{f}{1})$ on which $f$ is constant.
In other words, if a Boolean function has small spectral norm then it is
constant on a large affine subspace. 

The proof of the lemma uses a greedy algorithm that always makes the two largest Fourier coefficients 
to collide (with the same sign). 
Exploiting the property that $f$ is Boolean, 
one can show that this greedy folding either significantly increases the largest Fourier coefficient, 
or decreases the $\fn{f}{1}$ by a constant. 

The lemma immediately implies the following result for general (not necessarily XOR) functions.  
\begin{Thm}\label{thm:l1norm}
	For all $f:\B^m\times \B^n \to \B$, we have $\dcc(f) \leq 2\D_\oplus (f) = O(\deg_2(f)\cdot \fn{f}{1})$. 
\end{Thm}
In \cite{Gro97}, Grolmusz gave a public-coin randomized protocol with communication cost 
$O(\fn{f}{1}^2)$. The above theorem gives a \emph{deterministic} protocol, 
and the bound is better for functions $f$ with $\deg_2(f) = o(\fn{f}{1})$. 

Another implication of Lemma \ref{lem:l1norm} is that the communication complexity of 
$f\circ \oplus$ is at most the square root of the matrix rank.

\begin{Thm}\label{thm:sqrtl0}
For all $f:\BntB$,  
\[\dcc(f\circ \oplus) = O(\deg_2(f)\cdot \fn{f}{1}) = O\Big(\sqrt{\rank(M_{f\circ \oplus})}\log\rank(M_{f\circ \oplus})\Big).\]
\end{Thm}
The upper bound of $\rank/\log\rank$ in \cite{BLR12} improves the trivial linear bound 
by a log factor for all Boolean functions, 
assuming the Polynomial Freiman-Ruzsa conjecture. 
In comparison, our bound of $\sqrt{\rank}\log\rank$ is only for XOR functions, 
but it improves the linear bound by a polynomial factor, and it is unconditional.

It is also interesting to note that,
for any fixed Boolean function $f$,
at least one of the above two theorems gives a desirable result: 
either $\fn{f}{1} \geq \log^k\fn{f}{0}$ where $k$ is a big constant, 
then Theorem \ref{thm:l1norm} improves Grolmusz's bound almost quadratically 
(since $\deg_2(f)\leq \log\fn{f}{0}$); 
or $\fn{f}{1} \leq \log^k\fn{f}{0}$, 
then Theorem \ref{thm:sqrtl0} confirms the Log-rank conjecture for $f\circ \oplus$!  

\paragraph{Result 4: functions with a light Fourier tail.} 
Our last result deals with Boolean functions whose Fourier spectrum has a light tail. 
We call a function $f:\Bn\to\pmB$ $\mu$-close to $s$-sparse in $\ell_2$ if 
$\sum_{i>s} \hat f(s_i)^2 \leq \mu^2$, 
where $|\hat f(s_1)| \geq ... \geq |\hat f(s_N)|$. 
We say two functions $f,g:\Bn\to\pmB$ are $\epsilon$-close 
if $\pr_x[f(x) \neq g(x)] \leq \epsilon$. 
\begin{Thm}\label{thm:lighttail}
	If $f:\Bn\to\pmB$ is $\mu$-close to $s$-sparse in $\ell_2$, 
	where $\mu \leq \frac{\log^{O(1)}\fn{f}{0}}{\sqrt{\fn{f}{0}}}$ 
	and $s\leq \log^{O(1)}\fn{f}{0}$, 
	then $\D_\oplus(f) \leq \log^{O(1)}\fn{f}{0}$.
\end{Thm}


The proof of this theorem uses Chang's lemma about large Fourier coefficients of low-density functions, 
and a ``rounding'' lemma from \cite{GOS+11}.




\subsection{Related work}\label{sec:relatedwork}
The Log-rank Conjecture for XOR functions was shown to be true for symmetric functions~\cite{ZS09},
linear threshold functions (LTFs), monotone functions~\cite{MO10}, and $AC^0$ functions~\cite{KS13}.
These results fall into two categories. The first one, including symmetric functions and LTFs, is that the rank of the communication matrix (\ie the Fourier sparsity) is so large, that the Log-rank conjecture trivially holds. The second one, including monotone functions and $AC^0$ functions, is that even the Fourier degree is small, thus the standard decision tree complexity $\D(f)$ is already upper bounded by the poly-logarithmic of the matrix rank. But as we mentioned, there are functions that have small Fourier sparsity \emph{and} high Fourier degree (even after basis rotation), which form the hardcore cases of the problem. 
Our study makes crucial use the fact that parity queries are more powerful than single input-variable queries, and our results reveal structural properties of Fourier spectra.
 
In \cite{HS10}, Haramaty and Shpilka proved that 
$\rank(f) = O(\log^2 (1/\|f\|_{U^3})) = O(\log^2 (1/\bias(f)))$ for degree-$3$ polynomials. 
For degree-$4$ polynomials, however, 
the bound gets exponentially worse, and there were no results for higher degrees. 
In comparison, our Lemma \ref{lem:rank by 1norm} gives a polylog upper bound for 
$\rank(f)$ of all constant degree functions $f$, 
but the polylog is in $\fn{f}{1}$ rather than in $\bias(f)$ or Gower's norm (\cite{Gow98,Gow01,AKKLR05}). 

Though Boolean functions with a sparse Fourier spectrum 
seem to be a very interesting class of functions to study, 
not many properties are known. 
It is shown in~\cite{GOS+11} that the Fourier coefficients of a Fourier sparse function 
have large ``granularity'' and functions that are very close to Fourier sparse can be 
transformed into one through a ``rounding off'' procedure. 
Furthermore, they proved that one can use $2\log\fn{f}{0}$ random linear functions 
to partition the character space so that, with high probability, 
each bucket contains at most one nonzero Fourier coefficient. 
This does not help our problem since what we need is exactly the opposite: 
to group Fourier coefficients into buckets so that a small number of foldings would 
make many of them to collide (and thus reducing the Fourier sparsity quickly).

Let $A=\supp(\hat f)$ be the support of $f$'s Fourier spectrum.
One way of designing the parity query is to look for a ``heavy hitter'' 
$t$ of set $A+A$, \ie $t$ with many $s_1, s_2 \in A$ and $s_1 + s_2 = t$.
If such $t$ exists, then querying the linear function $\langle t, x\rangle$ reduces the Fourier sparsity a lot. 
One natural way to show the existence of a heavy hitter is by proving that $|A+A|$ is small. 
Turning this around, one may hope to show that if it is large, 
then the function is not Fourier sparse or has some special properties to be used. 
The size of $|A+A|$ has been extensively studied in additive combinatorics, 
but it seems that all related studies are concerned with the low-end case, in which 
$|A+A| \leq k|A|$ for very small (usually constant) $k$. Thus those results do not apply to our question.


There are actually two variants of polynomial rank. 
One is what we mentioned earlier and used in this work, and the other, which is
actually much better studied, is the minimum $r$ \st $f$ can be expressed as a function $F$ of $r$ 
lower degree polynomials $f_1, ..., f_r$. 
A nice result for this definition of rank is that 
large bias implies low polynomial rank \cite{GT09,KL08}: the rank is a function of the bias and degree only, 
but not of the input size $n$. 
This is, however, insufficient for us because a Fourier sparse function may have very small bias. 
Furthermore, the dependence of the rank on the degree is a very rapidly growing function 
(faster than a tower of 2's of height $d$), while our protocol has ``only'' single exponential dependence of $d$.  

\medskip
\paragraph{The work of~\cite{SV13}.}
After completing this work independently, the very recent work~\cite{SV13} came to our attention, 
which studies PDT complexity of functions with small spectral norm. 
The authors show $\C_{\oplus, \min}(f) \leq O(\fn{f}{1}^2)$ and $\D_\oplus(f) = O(\fn{f}{1}^2\log\fn{f}{0})$. 
In comparison, our Lemma \ref{lem:cmin vs L1} and Theorem \ref{thm:sqrtl0} are at least quadratically better. 
The paper \cite{SV13} also studies the size of PDT and shows that 
$\osize(f) \leq n^{O(\fn{f}{1}^2)}$,
and considers approximation of Boolean functions,
which is not studied in this paper.

\section{Preliminaries and notation}
All the logarithms in this paper are base 2. For two $n$-bit vectors $s,t\in \Bn$, 
define their inner product as $s\cdot t = \langle s, t\rangle = \sum_{i=1}^n s_it_i \text{ mod }2$
and for simplicity we write $s+t$ for $s\oplus t$.
Throughout the paper, logarithm is base $2$. 
We often use $f$ to denote a real function defined on $\Bn$. 
In most occurrences $f$ is a Boolean function, whose range can be represented by either $\B$ or $\pmB$, 
and we will specify whenever needed. 
For $f:\BntB$, we define $f^\pm = 1-2f$ to convert the range to $\pmB$. 
For each $b\in \range(f)$, the $b$-density of $f$ is $\rho_b=|f^{-1}(b)|/2^n$. 

Each Boolean function $f:\BntB$ can be viewed as a polynomial over \gf, 
and we use $\deg_2(f)$ to denote the \gf-degree of $f$. 
For a Boolean function $f:\BntB$ and a direction vector $t\in \Bn-\{0^n\}$, 
its derivative $\Delta_t f$ is defined by $\Delta_t f(x) = f(x) + f(x+t)$. 
It is easy to check that $\deg_2(\Delta_t f) < \deg_2(f)$ for any non-constant $f$ and any $t$. 

\paragraph{Complexity measures.} 
A parity decision tree (PDT) for a function $f:\BntB$ is a tree with each internal node 
associated with a linear function $\ell(x)$, and each leaf associated with an answer $a\in \B$. 
When we use a parity decision tree to compute a function $f$, 
we start from the root and follow a path down to a leaf. 
At each internal node, we query the associated linear function, 
and follow the branch according to the answer to the query. 
When reaching a leaf, we output the associated answer. 
The parity decision tree computes $f$ if on any input $x$, 
we always get the output equal to $f(x)$. The deterministic parity decision tree complexity of $f$, denoted by $\D_\oplus(f)$, is the least number of queries needed on a worst-case input by a PDT that computes $f$.

For a Boolean function $f$ and an input $x$, the parity certificate complexity of $f$ on $x$ is 
\[
	\C_\oplus(f,x) = \min\{\codim(H): x\in H, \text{$H$ is an affine subspace, on which $f$ is constant}\}.
\] 
The parity certificate complexity $\C_\oplus (f)$ of $f$ is $\max_x \C_\oplus(f,x)$. 
Since for each $x$ and each parity decision tree $T$, the leaf that $x$ belongs to corresponds to 
an affine subspace of co-dimension equal to the length of the path from it to the root, 
we have that $\C_\oplus(f) \leq \D_\oplus(f)$ \cite{ZS10}.
We can also study the minimum parity certificate complexities 
$\C_{\oplus,\min}^b(f) = \min_{x:f(x)=b} \C_\oplus(f,x)$ 
and 
$\C_{\oplus, \min}(f) = \min_x \C_\oplus(f,x)$.

Denote by $\dcc(F)$ the deterministic communication complexity of $F$. 
One way of designing communication protocols is to simulate a decision tree algorithm, 
and the following is an adapted variant of a well known relation between 
deterministic communication complexity and decision tree complexity 
to the setting of XOR functions and parity decision trees. 
\begin{Fact}\label{fact:CCbyDT}
	$\dcc(f\circ \oplus) \leq 2\D_\oplus(f)$.
\end{Fact}

\subsection*{Fourier analysis}
For any real function $f:\Bn\to\mbR$, 
the Fourier coefficients are defined by $\hat f(s) = 2^{-n}\sum_x f(x)\chi_s(x)$, 
where $\chi_s(x) = (-1)^{s\cdot x}$. 
The function $f$ can be written as 
$f = \sum_s \hat f(s) \chi_s$. 
The $\ell_p$-norm of $\hat f$ for any $p>0$, 
denoted by $\fn{f}{p}$, 
is defined as $(\sum_s |\hat f(s)|^p)^{1/p}$. 
The Fourier sparsity of $f$, denoted by $\|\hat f\|_0$, 
is the number of nonzero Fourier coefficients of $f$. 
The following is a simple consequence of Cauchy-Schwarz inequality 
\begin{equation}\label{eq:fourier_l0_l1}
\fn{f}{1}\leq \sqrt{\fn{f}{0}}.
\end{equation}
Note that $\fn{f}{1}$ can be much smaller than $\fn{f}{0}$. 
For instance, the AND function has $\fn{f}{1} \leq 3$ but $\fn{f}{0} = 2^n$. 
The Fourier coefficients of $f:\BntB$ and $f^\pm$ are related by 
$\widehat{f^\pm}(s) = \delta_{s,0^n} - 2 \hat f(s)$, 
where $\delta_{x,y}$ is the Kronecker delta function.
Therefore we have 
\begin{equation}\label{eq:range switch}
	2\fn{f}{1} - 1 \leq \wfn{f^\pm}{1} \leq 2\fn{f}{1} + 1,\quad 
\text{ and } \quad \fn{f}{0} - 1 \leq \wfn{f^\pm}{0} \leq \fn{f}{0} + 1. 
\end{equation}

For any function $f:\Bn\to\mbR$, 
Parseval's Indentity says that $\sum_s \hat f_s^2 = \av_{x}[f(x)^2]$. 
When the range of $f$ is $\B$, then $\sum_s \hat f_s^2 = \av_x [f(x)]$. 
We sometimes use $\hat f$ to denote the vector of $\{\hat f(s): s\in \Bn\}$. 

\begin{Prop}[Convolution]\label{prop:convolution}
For two functions $f,g:\Bn\to\mbR$, the Fourier spectrum 
of $fg$ is given by the following formula: 
$\widehat{fg}(s) = \sum_t \hat f(t)\hat g(s+t)$.
\end{Prop}
Using this proposition, one can characterize the Fourier coefficients of Boolean functions as follows.
\begin{Prop}\label{prop:Boolean}
A function $f: \Bn\to\mbR$ has range $\pmB$ if and only if 
\[
\sum_{t\in \Bn}\hat{f}^{2}(t)=1, 
\quad \text{and} \quad \sum_{t\in \Bn}\hat{f}(t)\hat{f}(s+t) = 0, 
\quad \forall s\in \Bn-0^n.
\]
\end{Prop}
Another fact easily follows from the convolution formula is the following. 
\begin{Lem}\label{lem:norm of product}
Let $f, g : \B^n \to \mbR$, then 
$\|\widehat{f g}\|_0 \leq \|\hat{f}\|_0 \|\hat{g}\|_0$ and 
$\|\widehat{f g}\|_1 \leq \|\hat{f}\|_1 \|\hat{g}\|_1$.
\end{Lem}

\paragraph{Linear maps and restrictions.}
Sometimes we need to rotate the input space: For an invertible linear map $L$ on $\Bn$, 
define $Lf$ by $Lf(x) = f(Lx)$. It is not hard to see that $\deg_2(Lf) = \deg_2(f)$, 
and that $\widehat{Lf}(s) = \hat f((L^T)^{-1}s)$. Thus 

\begin{equation}\label{eq:rotation norm}
	\wfn{Lf}{1} = \fn{f}{1} \text{ and }\wfn{Lf}{0} = \fn{f}{0}. 
\end{equation}

For a function $f:\Bn\to\mbR$, define two subfunctions $f_0$ and $f_1$, 
both on $\B^{n-1}$: $f_b(x_2, \ldots, x_n) = f(b,x_2, \ldots, x_n)$. 
It is easy to see that for any 
$s\in\B^{n-1}$, $\hat f_b(s) = \hat f(0s) + (-1)^b \hat f(1s)$, thus 
\begin{equation}\label{eq:subfn norm}
	\fn{f_b}{0} \leq \fn{f}{0} \text{ and }\fn{f_b}{1} \leq \fn{f}{1}.
\end{equation} 
The concept of subfunctions can be generalized to general directions. 
Suppose $f:\Bn\to\mbR$ and $S\subseteq \Bn$ is a subset of the domain. 
Then the restriction of $f$ on $S$, 
denoted by $f|_S$ is the function from $S$ to $\mbR$ defined naturally by $f|_S(x) = f(x)$, $\forall x\in S$. 
In this paper, we are concerned with restrictions on affine subspaces.

\begin{Lem}\label{lem:rotation}
Suppose $f:\Bn\to\mbR$ and $H = a+V$ is an affine subspace, 
then one can define the spectrum $\widehat{f|_H}$ of the restricted function $f|_H$ such that
\begin{enumerate}
\item If $\codim(H) = 1$, then $\widehat{f|_H}$ is the collection of 
  $\hat f(s) + (-1)^{b} \hat f(s+t)$ for all unordered pair $(s,s+t)$, 
  where $t$ is the unique non-zero vector orthogonal to $V$, and $b = 0$ if $a \in V$ and $b = 1$ otherwise. 
  Sometimes we refer to such restriction as a \emph{folding} over $t$.  
\item $\wfn{f|_H}{p} \leq \fn{f}{p}$, for any $p\in [0,1]$.
\item If $\range(f) = \pmB$, then the following three statements are equivalent: 
 1) $f|_H(x) = c \chi_s(x)$ for some $s\in \Bn$ and $c\in \pmB$, 
 2) $\wfn{f|_H}{0} = 1$, and 3) $\wfn{f|_H}{1} = 1$. 
\end{enumerate}
\end{Lem}
See Appendix~\ref{sec:restriction} for a proof. 
In the proof, we use a rotation $R$ as an isomorphism from $H$ 
(which may not be an group under addition any more) 
to the additive group of $\B^{n-1}$ (when $\codim(H) = 1$). 
Though the rotation is not unique, the resulting Fourier vector $\widehat{f|_H}$ is 
the same up to a linear invertible transform, thus the norm $\wfn{f|_H}{p}$ does not depend on the rotation. 
In addition, the \gf-degree of the subfunction $f|_H$ does not depend on the rotation $R$, 
thus we will just define $\deg_2(f|_H)$ to be the $\deg_2(f_b)$ where $f_b$ is the newly defined subfunctions. 

Using the above lemma, it is not hard to prove by induction the following fact, 
which says that short PDT gives Fourier sparsity. 
\begin{Prop}
For all $f:\BntB$, $\fn{f}{0} \leq 4^{\scriptsize \D_\oplus(f)}$.
\end{Prop}

The following theorem \cite{BC99} says that the \gf-degree can be bounded from above by logarithm of Fourier sparsity.
\begin{Fact}[\cite{BC99}]\label{fact:deg vs sparsity}
	For all $f:\BntB$, it holds that $\deg_2(f) \leq \log \fn{f}{0}$.
\end{Fact}

\section{Polynomial rank and the Main PDT algorithm}
The following notion of polynomial rank has been studied in \cite{Dic58} for degree-$2$ polynomials 
and in \cite{HS10} for degree-$3$ polynomials.\footnote{Degree-$4$ polynomials
was also studied in \cite{HS10}, but the rank is slightly different there as 
they allow some summands to be product of two quadratic polynomials.} 
\begin{Def}
The polynomial rank of a function $f\in \gf[x_1, \ldots, x_n]$, denoted $\rank(f)$, 
is the minimum integer $r$ \st $f$ can be expressed as 
\[
f = \ell_1f_1 + \ldots + \ell_rf_r + f_0,
\] 
where $\deg_2(\ell_i) = 1$ for all $1\leq i \leq r$ and 
$\deg_2(f_i) < \deg_2(f)$ for all $0 \leq i\leq r$. 
Sometimes we emphasize the degree by writing the polynomial rank as $\rank_d(f)$ with $d = \deg_2(f)$.
\end{Def}

Recall that a parity certificate is an affine subspace $H$ restricted on which $f$ is a constant. 
The parity certificate complexity is the largest co-dimension of such $H$. 
The next lemma says that the polynomial rank is quite small compared to the parity certificate complexity, 
even if we merely require $f$ to have a lower \gf-degree (rather than be constant) 
on the affine subspace; and in addition, 
even if we take the \emph{minimum} co-dimension over all such $H$.
\begin{Lem}\label{lem:rank-subspace}
For all non-constant $f:\BntB$, the following properties hold.
\begin{enumerate}
	\item There is a subspace $V$ of co-dimension $r = \rank(f)$ \st 
	when restricted to each of the $2^r$ affine subspaces $a+V$,
	$f$ has \gf-degree at most $\deg_2(f)-1$. 
	\item For all affine subspaces $H$ with $\codim(H) < \rank(f)$, $\deg_2(f|_H) = \deg_2(f)$. 
\end{enumerate}
\end{Lem}


\begin{proof} Fix a non-constant function $f:\BntB$.
\begin{enumerate}
\item Suppose $f = \ell_1f_1 + \ldots + \ell_rf_r + f_0$, 
	where $r = \rank(f)$. 
	Let $V = \{x: \ell_i(x) = 0, \forall i\in [r]\}$. 
	Then the conclusion holds by the definition of polynomial rank.
\item Suppose $H = a+V$ where $V$ is a subspace. Let $d = \deg_2(f)$, $r = \rank(f)$ and $k = \codim(V)$. 
    Let $s$ be the smallest integer \st 
		\begin{equation}\label{eq:rankproof}
		f(x) = \ell_1(x) f_1(\ell_2(x),\cdots,\ell_n(x)) + \ldots + \ell_s(x)f_s(\ell_{s+1}(x),...,\ell_n(x)) + f_0(\ell_{s+1}(x),...,\ell_n(x))
		\end{equation}
	for some linear functions $\ell_1(x), \ldots, \ell_k(x)$ \st when viewed as vectors, $span\{\ell_{1},...,\ell_k\} = V^\bot $, and some functions $f_i$ whose \gf-degree are all strictly smaller than $d$. 
	By the definition of the polynomial rank, 
	we have that $r\leq s$. 
	Since we assumed that $k < r$, it holds that $k < s$. Consider the function 
	\[f'(\ell_{k+2},\ldots, \ell_n) = \ell_{k+1}f_{k+1}(\ell_{k+2},\ldots, \ell_n) + \cdots + \ell_{s}f_{s}(\ell_{s+1},\ldots, \ell_n) + f_{0}(\ell_{s+1},\ldots, \ell_n).\] 
	Since $k < s$, there is at least one term other than the last $f_{0}(l_{s+1},\ldots, l_n)$. Now that $s$ is minimized, the function $f'$ has \gf-degree equal to $d$,	because otherwise $f'$ can be written as just one deg-$(d-1)$ function, thus the number of terms in Eq.\eqref{eq:rankproof} can be reduced. 
	Furthermore, we claim that $f$  
	restricted on the affine subspace $V+a$ has \gf-degree equal to $d$ as well. 
	Indeed, the first $k$ terms in Eq.\eqref{eq:rankproof} give 
	$\ell_1 f_1(\ell_2,...,\ell_n) + \ldots + \ell_kf_k(\ell_{k+1},...,\ell_n)$, 
	which has \gf-degree at most $d-1$ after $\ell_1$, ..., $\ell_k$ take specific values given by $a$. 
	Thus it cannot cancel any degree-$d$ monomial in $f'$. 
	So $f$ on $V+a$ has \gf-degree $d$. 
\end{enumerate}
\end{proof}

Lemma~\ref{lem:rank-subspace}, though seemingly simple, 
is of fundamental importance to our problem as well as PDT algorithm designing in general.
Note that the second part of  Lemma~\ref{lem:rank-subspace} says that,
if there exists an affine subspace $V+a$ of co-dimension $k$ and a vector $a \in V^{\perp}$
such that $\deg_2(f|_{V+a})<\deg_2(f)$, then $\rank(f) \leq k$.
Therefore Lemma~\ref{lem:rank-subspace} reduces the challenging task of lowering the
degree of $f|_{V + a}$ for \emph{all} $a$ to lowering it for just \emph{one} $a$.

In the next two sections, what we are going to use is the following corollary of it. 
\begin{Cor}\label{cor:Cmin vs rank}
	For all non-constant $f:\BntB$, we have $\rank(f) \leq \C_{\oplus, \min}(f)$.
\end{Cor}
\begin{proof}
	This immediately follows from the second item of Lemma \ref{lem:rank-subspace}, 
	because $\C_{\oplus, \min}(f)$ requires $\deg_2(f|_H) = 0$, 
	strictly smaller than $\deg_2(f)$ for non-constant $f$.
\end{proof}

\subsection{Main PDT algorithm}
Now we describe the main algorithm for computing function $f$,
by reducing the \gf-degree of $f$.
\begin{center}
\fbox{
\begin{minipage}[l1pt]{6.00in}
{\bf Main PDT Algorithm} \\
{\bf Input}: An PDT oracle for $x$ \\
{\bf Output}: $f(x)$.
\begin{enumerate}
	\item \textbf{while} $\deg_2(f) \geq 1$ \textbf{do} 
	\begin{enumerate}
		\item Take a fixed decomposition $f = \ell_1 f_1 + \cdots + \ell_r f_r + f_0$, 
		   where $r = \rank_d(f)$ with $d = \deg_2(f)$.
		\item for $i = 1$ to $r$
		\item \quad Query $\ell_i(x)$ and get answer $a_i$. 
		\item Update the function 
		\[f := a_1 f_1 + \cdots + a_r f_r + f_0\]
	\end{enumerate}
\end{enumerate}
%
\end{minipage}
}
\end{center}

To analyze the query complexity of this algorithm, we need to bound $\rank(f)$.
We conjecture that the following is true for all Fourier sparse Boolean functions.
\begin{Conj}\label{conj:rank vs sparsity}
	For all Boolean functions $f:\BntB$, $\rank(f) = O(\log^c(\fn{f}{0}))$ for some $c=O(1)$. 
\end{Conj}
Call a complexity measure $M(f)$ \emph{downward non-increasing} if $M(f') \leq M(f)$ for any $f$ and any subfunction $f'$ of $f$. As mentioned earlier (Lemma \ref{lem:rotation}), $M(f) = \fn{f}{0}$ and $M(f) = \fn{f}{1}$ are all downward non-increasing complexity measures. 
\begin{Thm}\label{thm:mainPDT}
The Main PDT algorithm computes $f(x)$ correctly. 
If $\rank(f) \leq M(f)$ for some downward non-increasing complexity measure $M$, 
then $\D_\oplus(f) \leq \deg_2(f)M(f)$ and $\dcc(f\circ \oplus) \leq 2\log \fn{f}{0} \cdot M(f)$. 
In particular, if Conjecture \ref{conj:rank vs sparsity} is true, 
then the Log-rank conjecture holds for all XOR functions. 
\end{Thm}
\begin{proof}
The correctness is obvious. 
For the query cost, there are at most $\deg_2(f)$ rounds since each round 
reduces the \gf-degree by at least one. 
To avoid confusion, denote the original function by $f$ and the function in the iteration $t$ by $f^{(t)}$. Note that $f^{(t)}$ is obtained from $f$ by a sequence of linear restrictions, it is a subfunction of $f$. 
Each iteration $t$ takes $\rank(f)$ queries. 
If $\rank(f) \leq M(f)$, 
then in particular $\rank(f^{(t)}) \leq M(f)$ since $M$ is a downward non-increasing complexity measure. 
Taking all iterations together, the total number of queries is at most 
$\deg_2(f)M(f)$. The communication complexity $\dcc(f\circ \oplus) \leq 2\log \fn{f}{0} \cdot M(f)$ follows from the standard simulation result (Fact \ref{fact:CCbyDT}) and the degree bound (Fact \ref{fact:deg vs sparsity}).

If Conjecture \ref{conj:rank vs sparsity} is true, then the measure $M(f)$ is replaced by $\log^c(\fn{f}{0})$, and thus the above bound becomes $\dcc(f\circ \oplus) \leq 2\log^{c+1}(\fn{f}{0})$. Namely the Log-rank conjecture holds for all XOR functions.
\end{proof}

The Main PDT algorithm, though simple, crucially uses the fact that restrictions do not increase the 
Fourier sparsity and uses the \gf-degree as a progress measure to govern the efficiency. Since $\deg_2(f) \leq \log(\fn{f}{0})$, the algorithms finishes in a small number of rounds. 

This algorithm also gives a unified way to construct parity decision tree, reducing the task of designing PDT algorithms to showing that the polynomial rank is small. Indeed, the results in the next two sections are obtained by bounding rank, where sometimes Theorem \ref{thm:mainPDT} will be applied with the complexity measure $\fn{f}{1}$. 

Note that if the conjecture $\D_\oplus(f) \leq \log^c\fn{f}{0}$ is true, then the Main PDT algorithm always gives the optimal query cost up to a polynomial of power $c+1$.

\section{Functions with low \gf-degree}\label{sec:constant_degree}
In this section, we mainly show that the Log-rank conjecture holds for XOR functions with constant \gf-degree. We will actually prove 
\[\C_{\oplus,\min}(f) = O(2^{d^2/2} \log^{d-2} \fn{f}{1}),\] 
which is stronger than Lemma \ref{lem:rank by 1norm}. Theorem \ref{thm:constdeg} then follows from the PDT algorithm and the simulation protocol for PDT (Theorem \ref{thm:mainPDT}). Corollary \ref{cor:junta} is also easily proven by an induction on \gf-degree. 

The case for degree $1$ (linear functions) is trivial and 
the case for degree $2$ (quadratic polynomials) is also simple due to the following Dickson's theorem. 

\begin{Thm}[\cite{Dic58}]\label{thm:Dickson}
Let $A\in \B^{n\times n}$ be a symmetric matrix whose diagonal entries are all $0$, 
and define a polynomial $f(x) = x^T Q x + \ell(x) + \epsilon$, 
where $Q$ is the upper triangle part of $A$. 
Then $\rank(f)$ is equal to the rank of matrix $A$ over \gf.
\end{Thm}
Note that Dickson's theorem says that, up to an affine (invertible) linear map, 
the Fourier spectrum of a degree $2$ polynomial is identical to a bent function on $k$ variables
$f(x)=x_1x_2+\cdots +x_{k-1}x_k$,
where $k=\rank(A)$. Note that $\C_{\oplus,\min}(f) \leq k/2$ because we can simply fix $x_1 = x_3 = \ldots = x_{k-1} = 0$ and get a 0-constant function. 
It is also easily seen that this bent function has spectral norm $2^{k/2}$, it follows that
$\rank(f) \leq \C_{\oplus,\min}(f) = O(\log \fn{f}{1})$. 

\subsection{Cubic polynomials}
We prove Theorem~\ref{thm:constdeg} for the special case of cubic polynomials first.
This is because degree $3$ is the first non-trivial case and we use this result
in our final induction proof of Theorem~\ref{thm:constdeg}; 
more importantly, the proof applies some ideas from~\cite{HS10} which inspire our proof 
for the general constant degree case.

In \cite{HS10}, it was shown that for polynomials with \gf-degree $3$, 
$\rank(f) = O(\log^2 (1/\bias(f)))$. 
By shifting the Fourier spectrum appropriately, we can make $\bias(f) \geq 1/\sqrt{\fn{f}{0}}$ and 
thus get $\rank(f) = O(\log^2 \fn{f}{0})$. 
Next, we show that actually the bound can be improved to $\rank(f) = O(\log \fn{f}{0})$. 
This will also be used for the general degree case. 

We need a lemma that relates the rank of a cubic polynomial and the ranks of its derivatives. 
We call a function \emph{linear} if its \gf-degree is \emph{at most} 1, 
and \emph{quadratic} if its \gf-degree is \emph{at most} 2. 
The following statement is slightly more general than Lemma 3.7 in \cite{HS10}, 
but the same proof goes through. 
\begin{Lem}[\cite{HS10}]\label{lem:cubic rank}
Let $M$ be a collection of quadratic functions satisfying that $\rank_2(f) \leq r$ 
for all $f\in M\cup 2M$ (where $2M = \{f_1+f_2: f_1,f_2\in M\}$),	
then there is a subspace $V$ of co-dimension at most $4r$ \st $f|_V$ is a linear function for all $f\in M$.
\end{Lem}

Now we can prove the cubic polynomial case of Theorem~\ref{thm:constdeg}. 
\begin{Prop}\label{prop:deg3}
	For all function $f:\BntB$ with \gf-degree $3$, 
	it holds that $\rank(f) = O(\log \fn{f}{1})$ and thus $\D_\oplus(f) = O(\log\fn{f}{1})$.
\end{Prop}

\begin{proof}
Note that $\Delta_t f$ has \gf-degree at most 2 for all $t$, 
and that in general 
$\Delta_t f + \Delta_s f = \Delta_{t+s} f + \Delta_t \Delta_s f$. 
Let $M$ be the collection of $\{\Delta_t f: t\in \Bn\}$, 
then $M$ satisfies the condition of Lemma \ref{lem:cubic rank}. 
Furthermore, each $\Delta_t f\in M$ has 
\[
\rank(\Delta_t f) = \log\wfn{\Delta_t f}{1}+1 \leq 2\log\fn{f}{1}+1,
\] 
where the last inequality is because of Lemma~\ref{lem:norm of product}. 
Let $r = 2\log\fn{f}{1}+1$. Now by Lemma \ref{lem:cubic rank}, 
we know that $4r$ restrictions can make all $\Delta_t f$ in $M$ to become linear functions. 
Therefore there is a subspace of co-dimension at most $4r$ restricted on which 
$\Delta_t f$ are linear functions, for all $t\in \Bn$. 
This means that $f|_V$ has degree at most $2$.  
It follows that $\rank(f) \leq 4r$. 
The upper bound on $\D_\oplus(f)$ now follows by observing that 
$\D_\oplus(f) \leq \rank(f) + \D_\oplus(f')$, 
where $f'$ is a subfunction of $f$ with \gf-degree $2$. 
Recall that for subfunctions we have 
$\|\hat{f'}\|_1 \leq \|\hat{f}\|_1$, 
and hence $\D_\oplus(f') = O(\log\|\hat{f}\|_1+1) = O(\log\|\hat{f}\|_1)$.
\end{proof}

\subsection{Constant-degree polynomials}
Now we will bound the $\rank(f)$ and use the Main PDT algorithm to bound the PDT complexity. 
\begin{Lem}\label{lem:rank by 1norm full}
For all non-constant function $f:\Bn\to\B$ of \gf-degree $d$, we have 
\[
	\rank(f) \leq \C_{\oplus, \min}(f) \leq \D_\oplus (f) \leq O(2^{d^2/2} (\log^{d-2}\|\hat{f}\|_1+1)).
\] 
\end{Lem}
\begin{proof}
We will prove by induction on degree $d$ that 
\[\rank(f) \leq \C_{\oplus, \min}(f) \leq \max_{b\in \B} \C_{\oplus, \min}^b(f) 
\leq \D_\oplus (f) 
\leq B_d\big(\wfn{f^\pm}{1}\big),\] 
where $B_d(m) < 2^{d^2/2} \log^{d-2}m$ are 
a class of bounded non-decreasing (with respect to both $d$ and argument $m$) functions
to be determined later. 
The conclusion then follows from Eq.\eqref{eq:range switch}. 
The case of $d=1$ is trivial, the case $d=2$ is easily handled by Theorem \ref{thm:Dickson}, 
and the case $d = 3$ is given by Proposition~\ref{prop:deg3}. 
	
Now suppose that the bound holds for all polynomials of \gf-degree at most $d-1$, 
and consider a function $f$ of degree $d \geq 4$. 
We will first prove a bound for $\C_{\oplus, \min}(f)$, 
which also implies a bound on $\rank(f)$ from above by Corollary \ref{cor:Cmin vs rank}.
	
First, it is not hard to see that there exists a direction $t\in \Bn-\{0^n\}$ such that 
$\Delta_t f$ is non-constant (unless $f$ is a linear function, in which case the conclusion trivially holds anyway). 
Fix such a $t$.
%
Since $\deg_2(\Delta_t f) \leq d-1$, by induction hypothesis, it holds that
\[
\C_{\oplus, \min}^b(\Delta_t f) \leq  B_{d-1}(\wfn{(\Delta_t f)^\pm}{1}). 
\]
Define $f_t(x) = f(x+t)$, then by Lemma~\ref{lem:norm of product}, we have
\[
\|\widehat{(\Delta_t f)^\pm}\|_1 
= \|\widehat{f^\pm \cdot f_t^\pm}\|_1 
\leq \|\widehat{f^\pm}\|_1 \|\widehat{f_t^\pm}\|_1 
= \|\widehat{f^\pm}\|_1^2,
\] 
which implies that
\[\C_{\oplus, \min}^b(\Delta_t f)\leq B_{d-1}(\wfn{f^\pm}{1}^2).\] 
Since $\Delta_t f$ is non-constant, $C_{\oplus, \min}^{b}(\Delta_t f) \in [1,n]$ for \emph{both} 
$b = 0$ and $b = 1$. 
For each $b$, by the definition of $C_{\oplus, \min}^b(\Delta_t f)$, 
there exists an affine subspace $H_b$ with $\codim(H_b) \leq B_{d-1}(\wfn{f^\pm}{1}^2)$ 
such that $(\Delta_t f) |_{H_b} = b$, 
which is equivalent to $f(x) + f(x+t) = b$ for all $x \in H_b$. 

Define
\[
g_0(x) = \frac{1}{2}\big(f^\pm(x) +_{\mathbb R} f^\pm(x+t)\big), 
\qquad 
g_1(x) = \frac{1}{2}\big(f^\pm(x) -_{\mathbb R} f^\pm(x+t)\big),
\]
where the plus $+_{\mathbb R}$ and minus $-_{\mathbb R}$ are over $\mbR$. 
(To avoid potential confusions, in the rest of the proof, 
we will also use this notation for addition/subtraction of two functions over $\mbR$.) 

These two functions have some nice properties. 
First, it is easy to see from the definition of $g_0$ and $g_1$ that $f^\pm = g_0 +_{\mathbb R} g_1$. 
Second, note that $g_0$ and $g_1$ are not Boolean functions any more; 
they take values in $\{-1,0,+1\}$. 
However, a simple but crucial fact is that they take very special values on the affine subspace $H_b$: 
one always takes value 0, and the other always takes value in $\pmB$. 
Actually, it is not hard to verify that 
\[
g_b|_{H_b} = f^\pm|_{H_b}\quad \text{ and } \quad g_{1-b}|_{H_b} = 0.
\] 
Third, in the Fourier domain, note that 
\[
  \widehat{f^\pm_t}(s) 
= \av_x [f^\pm(x+t)\chi_s(x)] 
= \av_x [f^\pm(x+t)\chi_s(x+t)\chi_s(t)] 
= \widehat{f^\pm}(s) \chi_s(t),
\] 
and thus 
\[
 \widehat{g_b}(s) 
= \frac{1}{2}\big(\widehat{f^\pm}(s) +_{\mathbb R} (-1)^b\widehat{f^\pm_t}(s)\big) 
= \frac{1}{2}\big(\widehat{f^\pm}(s) +_{\mathbb R} (-1)^b\chi_s(t)\widehat{f^\pm}(s)\big).
\] 
Therefore, we have 
\[
 \widehat{g_0}(s) 
= \begin{cases} 
\widehat{f^\pm}(s) & s\in t^\bot\\ 
0 & s\in \overline{t^\bot}
\end{cases},
\quad \text{ and } \quad 
 \widehat{g_1}(s) 
= \begin{cases} 
0 & s\in t^\bot\\ 
\widehat{f^\pm}(s) & s\in \overline{t^\bot}
\end{cases},
\] 
where $t^\bot = \{s\in \Bn: \langle s, t \rangle = 0\}$. 
Namely $\hat{g}_0$ and $\hat{g}_1$ each takes the Fourier spectrum $f^\pm$ on 
one of the two hyperplanes defined by the vector $t$. 

This further implies that	
\[
	\|\widehat{f^\pm}\|_1 = \|\widehat{g_0} \|_1 + \|\widehat{g_1}\|_1.
\]
Thus, either $\|\widehat{g_0}\|_1$ or $\|\widehat{g_{1}}\|_1$ is at most half of $\|\widehat{f^\pm}\|_1$. 
Suppose that $\|\widehat{g_{b}}\|_1 \leq \frac{1}{2}\|\widehat{f^\pm}\|_1$. 
We claim that restricting $f^\pm$ to $H_b$ reduces its spectral norm a lot. 
Indeed, since $f^\pm|_{H_b} = g_b|_{H_b}$, 
we have 
\[
  \|\widehat{f^\pm|_{H_b}}\|_1 
= \|\widehat{g_{b}|_{H_b}}\|_1 \leq \|\widehat{g_{b}}\|_1 
\leq \frac{1}{2}\|\hat{f}\|_1,
\] 
where the first inequality is because of Lemma~\ref{lem:rotation}. 
To summarize, we have just shown that we can reduce the spectral norm by at least half 
using at most $B_{d-1}(\wfn{f^\pm}{1}^2)$ linear restrictions. 

Now we recursively repeat the above process on the subfunction $f^\pm|_{H_b}$ 
until finally we find an affine subspace $H$ \st 
$\|\widehat{f^\pm|_{H}}\| \leq 1$, 
at which moment the subfunction is either a constant or linear function, 
thus at most one more folding would give a constant function. 
In total it takes at most $B_{d-1}(\wfn{f^\pm}{1}^2) \log \|\widehat{f^\pm}\|_1 + 1$ 
linear restrictions to get a constant function, which implies that 
\[
\C_{\oplus, \min}(f) 
\leq B_{d-1}(\wfn{f^\pm}{1}^2) \log \|\widehat{f^\pm}\|_1 + 1.
\] 
	
Next we will show that actually the $\max_{b\in \B} \C_{\oplus, \min}^b(f)$ is not much larger either:
\begin{equation}
	\max_{b\in \B} \C_{\oplus, \min}^b(f) \leq B_{d-1}(\wfn{f^\pm}{1}^2) \log \|\widehat{f^\pm}\|_1 + B_{d-1}\big(\wfn{f^\pm}{1}\big) + 1.
\label{eq:maxminPC}
\end{equation}
(We need to show this because in the induction step, 
we picked one $g_b$ with smaller spectral norm and used the induction hypothesis to upper bound 
$\C_{\oplus, \min}^b(\Delta_t f)$ for a particular $b$, 
which could be $\max_{b\in \B} \C_{\oplus, \min}^b(\Delta_t f)$.) 
Note that by the Main PDT algorithm, we know that 
\[
\D_{\oplus}(f) \leq \rank(f) + \D_{\oplus}(f'),
\] 
for a subfunction $f'$ of $f$ with $\deg_2(f') < \deg_2(f)$. 
Now by Corollary \ref{cor:Cmin vs rank}, 
we can use $\C_{\oplus, \min}(f)$ to upper bound $\rank(f)$. 
For the second part, 
since $\deg_2(f') < \deg_2(f)$ 
and $\wfn{(f')^\pm}{1} \leq \wfn{f^\pm}{1}$, 
we can apply the induction hypothesis on $f'$ to upper bound $\D_{\oplus}(f')$. 
What we get here is 
\begin{equation}
	\D_\oplus (f) \leq 
B_{d-1}\big(\wfn{f^\pm}{1}^2\big) \log \|\widehat{f^\pm}\|_1 + 1 + B_{d-1}\big(\wfn{f^\pm}{1}\big).
\label{eq:PDTbound}
\end{equation}
Eq.\eqref{eq:maxminPC} thus follows from the simple bound $\C_{\oplus, \min}^b(f) \leq \C_{\oplus}(f) \leq \D_{\oplus}(f)$. Now define the right-hand side of Eq.\eqref{eq:PDTbound} to be $B_d \big(\wfn{f^\pm}{1}\big)$, and solve the recursive relation 
\[
B_d(m) = B_{d-1}(m^2)\log m + B_{d-1}(m) + 1, \quad B_3(m) = O(\log m + 1),
\] 
we get 
\[
B_d(m) = (1+o(1))2^{(d-2)(d-3)/2} \log^{d-2} m,
\]
as desired.
\end{proof}

Note that in the above proof, it seems that we lose something by using 
$\C_{\oplus, \min}$ to upper bound $\rank$. 
However, it is crucial to consider the affine subspace $H_b$ on which 
$\Delta_t f$ becomes a \emph{constant} (instead of, say, a lower \gf-degree polynomial), 
because otherwise $g_{b}$ on $H_b$ is not equal to $f$ 
(actually not even Boolean), and thus we cannot recursively apply the procedure on $f|_{H_b}$. 
In addition, if $\Delta_t f$ is not constant on $H_b$, then we cannot guarantee the decrease of the spectral norm due to restriction on $H_b$. 

We have just showed that low degree polynomials have very small $\C_{\oplus,\min}$ value in terms of the spectral norm. 
We actually conjecture that the bound can be improved to the following.
\begin{Conj}\label{conj:cmin by l1}
There is some absolute constant $c$ \st for any non-constant 
$f:\BntB$, $\C_{\oplus,\min}(f) = O(\log^c\fn{f}{1})$.
\end{Conj}
It has the consequence as follows.
\begin{Prop}\label{prop:rankL1}
If Conjecture \ref{conj:cmin by l1} is true, 
then for any $f:\BntB$, $\rank(f) = O(\log^c\fn{f}{1})$ and 
$\D_\oplus(f) = O(\deg_2(f) \log^c\fn{f}{1})$.
\end{Prop}

In fact, we are not aware of any counterexample for Conjecture \ref{conj:cmin by l1} even for $c=1$; 
see the last section for more discussions on this.

Lemma \ref{lem:rank by 1norm full} also implies the following Corollary, from which Corollary \ref{Cor:constdegGS} immediately follows. 
\begin{Cor}\label{cor:dGS}
	If $f:\Bn\to\B$ has \gf-degree $d$, then $f = \sum_{i=1}^{T} \pm \mathbbm{1}_{V_i}$, where $T = 2^{2^{d^2/2}\log^{d-2}\fn{f}{1}}$ and each $\mathbbm{1}_{V_i}$ is the indicator function of the subspace $V_i$.
\end{Cor}
\begin{proof}
	By Lemma \ref{lem:rank by 1norm full}, we know that the depth of the optimal PDT is at most $2^{d^2/2}\log^{d-2}\fn{f}{1}$, and thus the size of the PDT is at most $T$. So the function can be written as the sum of at most $T$ indicator functions $\mathbbm{1}_H$ of affine subspaces. Then as argued in \cite{GS08}, each such indicator $\mathbbm{1}_H$ can be written as $\mathbbm{1}_{V_1} - \mathbbm{1}_{V_2}$ for two subspaces $V_1$ and $V_2$. The conclusion thus follows.
\end{proof}

\section{Functions with a small spectral norm}\label{sec:L1}
We prove Lemma~\ref{lem:cmin vs L1} in this section,
which directly implies Lemma~\ref{lem:l1norm}, Theorem~\ref{thm:l1norm} and Theorem \ref{thm:sqrtl0}.

\medskip
\begin{Lem}\label{lem:cmin vs L1}
	For all Boolean function $f:\Bn\to\pmB$, we have $\C_{\oplus, \min}(f) \leq O(\fn{f}{1})$.
\end{Lem}
\begin{proof}
Suppose that the nonzero Fourier coefficients are $\{\hat f(\alpha): \alpha\in A\}$, 
where $A = \supp(\hat f)$. Denote by $a_1, a_2, ..., a_s$ the sequence of 
$|\hat f(\alpha)|$ in the decreasing order, 
and the corresponding characters are $\chi_{\alpha_1}, ..., \chi_{\alpha_s}$ 
in that order (thus $|\hat f(\alpha_i)| = a_i$ and $s=\fn{f}{0}$ is the Fourier sparsity of $f$). 
For simplicity, we assume $s\geq 4$, as doing so can only add at most a constant to 
our bound on $\C_{\oplus, \min}(f)$.
 
Consider the following \emph{greedy folding} process: 
fold along $\beta = \alpha_1 + \alpha_2$ and select a proper half-space, 
namely impose a linear restriction $\chi_{\beta}(x) = b$ for some $b\in \B$, \st the subfunction 
has its largest Fourier coefficient being $a_1 + a_2$ (in absolute value). 
This is achievable according to Lemma \ref{lem:rotation}.

We first show that at most $O(\|\hat{f}\|_1)$ greedy foldings can boost
$a_1$, the largest Fourier coefficient in absolute value, to at least $1/2$.
By Parseval's Identity, we have 
\[1-a_1^2 = \sum_{i \geq 2} a_i^2 \leq a_2 \sum_{i \geq 2} a_i = a_2 (\|\hat{f}\|_1 - a_1).\]
So when $a_1 \leq 1/2$, the greedy folding increases the largest coefficient by 
\[a_2 \geq \frac{1-a_1^2}{\|\hat{f}\|_1 - a_1} > \frac{3}{4\|\hat{f}\|_1}.\] 
Hence the largest coefficients would be larger than $1/2$ in $O(\|\hat{f}\|_1)$ steps. 
(After one folding, the function becomes a subfunction of the previous one, 
but due to Lemma \ref{lem:rotation}, 
the $\ell_1$ norm of its Fourier spectrum only decreases. 
So we can safely use $\fn{f}{1}$ as a universal upper bound for this sequence of subfunctions.)
	
Next we show that greedy folding decreases the Fourier $\ell_1$-norm by at least 
$2a_1 = 2\max_s |\hat f(s)|$. 
Define 
\[P_+(\beta) = \{(s,t): \text{$s+t=\beta$, $\hat f(s)\cdot \hat f(t) > 0$}\} \quad \text{and} \quad
P_-(\beta) = \{(s,t): \text{$s+t=\beta$, $\hat f(s)\cdot \hat f(t) < 0$}\},\]
where all pairs $(s,t)$ are unordered; same for the rest of the proof. 
By comparing the old and new Fourier spectra, we can easily see that the drop of Fourier $\ell_1$-norm 
is precisely
\[
2\cdot \sum_{(\alpha_i,\alpha_j)\in P_-(\beta)} \min\{|\hat f(\alpha_i)|, |\hat f(\alpha_j)|\}.
\] 
Note that the folding is chosen such that the largest two Fourier coefficients have the same sign, 
so $(\alpha_1, \alpha_2)\in P_+(\beta)$. 
Next we will use the property that $f$ is a Boolean function. By Proposition~\ref{prop:Boolean},
 $\sum_{\alpha_i+\alpha_j = \beta} \hat f(\alpha_i) \hat f(\alpha_j)= 0$, 
 thus 
 $\sum_{(\alpha_i,\alpha_j) \in P_+(\beta)} a_ia_j  = \sum_{(\alpha_i,\alpha_j) \in P_-(\beta)} a_i a_j$. 
 Now we have 
\[
  a_1 a_2 
\leq \sum_{(i,j) \in P_+(\beta)} a_i a_j 
= \sum_{(i,j) \in P_-(\beta)} a_i a_j 
\leq a_3 \sum_{(i,j) \in P_-(\beta)} \min\{a_i, a_j\}.
\]
Therefore, the decrease of the Fourier $\ell_1$-norm is at least 
$\frac{2a_1a_2}{a_3} \geq 2a_1$. 
Thus once $a_1 > 1/2$, then each greedy folding decreases the Fourier $\ell_1$-norm by at least 1. 
So it takes at most $\fn{f}{1}$ further steps to make the Fourier $\ell_1$-norm to be at most 1, 
in which case at most one more folding makes the function constant. 
\end{proof}

Lemma~\ref{lem:cmin vs L1} implies that $\rank(f) \leq O(\fn{f}{1})$ (Lemma~\ref{lem:l1norm})
by Corollary~\ref{cor:Cmin vs rank} (that $\rank(f) \leq \C_{\oplus, \min}(f)$).  

Note that our Main PDT algorithm can be simply simulated by a protocol in which 
\alice and \bob send $\ell_i(x)$ and $\ell_i(y)$, respectively. 
Thus, similar to Fact \ref{fact:CCbyDT}, we have $\dcc(f)\leq 2\D_\oplus(f)$ for 
$f:\Bn\times\B^m\to\B$. Theorem \ref{thm:l1norm} basically follows from this lemma and the fact 
that subfunctions have smaller spectral norm (Lemma~\ref{lem:rotation}). 

Lemma~\ref{lem:cmin vs L1} also implies Theorem \ref{thm:sqrtl0}, which asserts 
upper bounds on the deterministic communication complexity of $f\circ \oplus$ as
\[
\dcc(f\circ \oplus) = O(\deg_2(f)\cdot \fn{f}{1}) = 
O\Big(\sqrt{\rank(M_{f\circ \oplus})}\log\rank(M_{f\circ \oplus})\Big).
\]
To see this, first recall Theorem~\ref{thm:mainPDT},
which states that $\dcc(f\circ \oplus) \leq 2\log \fn{f}{0} \cdot M(f)$ 
where $M(f)$ is a downward non-increasing complexity measure.
By Lemma~\ref{lem:l1norm}, we can take $M$ to be $\fn{f}{1}$.
Now combining these with 
Fact~\ref{fact:deg vs sparsity} (that $\deg_2(f) \leq \log \fn{f}{0}$), 
and the inequality that $\fn{f}{1} \leq \sqrt{\fn{f}{0}}$
yields Theorem \ref{thm:sqrtl0}.

\section{Functions with a light Fourier tail}
First we will show that functions with low density can be computed efficiently by PDT. 
We will need a result by Chang \cite{Cha02}. 
The following version is taken from a simplified proof in \cite{IMR12}. 
Recall that for a function $f:\BntB$, its density is $\rho_1(f) = |f^{-1}(1)|/2^n$.
\begin{Lem}[\cite{Cha02,IMR12}]\label{lem:Chang}
For all $f:\BntB$ and any $\epsilon > 0$, 
the set $\{s: |\hat f(s)|\geq \epsilon\}$ spans a subspace of dimension less than 
$d = 2\big(\frac{\rho_1(f)}{\epsilon}\big)^2 \ln(1/\rho_1(f))$.
\end{Lem}
Another fact that we will need is the granularity of Boolean functions, first studied in \cite{GOS+11}. 
\begin{Def}
The (Fourier) granularity of a function $f:\Bn\to\pmB$, 
denoted $\gran(f)$, 
is the minimum integer $k$ \st all nonzero Fourier coefficients are integer multiples of $2^{-k}$. 
\end{Def}
The following theorem relates granularity and sparsity.
\begin{Lem}[\cite{GOS+11}] \label{lem: spar-gran}
Any Boolean function $f:\Bn\to\pmB$ with $\fn{f}{0} \geq 2$, 
has $\gran(f) \leq \lfloor \log_2 \fn{f}{0} \rfloor -1$.	
\end{Lem}

Now we can show the lemma for low-density functions. 
\begin{Lem}\label{lem:sparse}
For all $f:\BntB$ with $\rho_1(f) = \frac{\polylog(\fn{f}{0})}{\fn{f}{0}}$, 
$\D_\oplus(f) \leq \log^{O(1)}\big(\fn{f}{0}\big)$.
Thus the Log-rank Conjecture is true for $f\circ \oplus$. 
\end{Lem}
For completeness, we give a self-contained proof (without resorting to~\cite{IMR12}) 
of this lemma using Beckner-Bonami 
inequality in Appendix~\ref{sec:BB_proof} with a slightly worse parameter.
\begin{proof}
Suppose that $\rho_1(f) = \log^c \fn{f}{0}/\fn{f}{0}$. 
By Lemma~\ref{lem: spar-gran}, the minimum Fourier coefficient (in absolute value) 
is at least $2/\fn{f}{0}$. 
Take $\epsilon$ as this value, and apply Lemma \ref{lem:Chang}, 
we know that all nonzero Fourier coefficients are in a subspace of dimension 
\[
	O((\rho_1(f)\fn{f}{0})^2\ln(1/\rho_1(f))) = O(\log^{2c+1} \fn{f}{0}).
\]
This implies that there exists an invertible linear transformation $L$ such that all the non-zero Fourier
coefficients of $f \circ L$ lie in a subspace of dimension $d = O(\log^{2c+1} \fn{f}{0})$. 
By choosing the basis appropriately, we may assume, without loss of generality, that the subspace is just 
$\{0, 1\}^d\times 0^{n-d}$. 
Thus a decision tree algorithm for $f\circ L$ can simply query these bits. 
Therefore, $\D_\oplus(f) \leq \D(f\circ L) \leq \log^{O(1)}\big(\fn{f}{0}\big)$.
\end{proof}

The last lemma we need is the following result by Gopalan \etal \cite{GOS+11}. 
Recall that a function $f:\Bn\to\pmB$ is $\mu$-close to $s$-sparse in $\ell_2$ if 
$\sum_{i>s} \hat f(s_i)^2 \leq \mu^2$, 
where $|\hat f(s_1)| \geq ... \geq |\hat f(s_N)|$. 
We say two functions $f,g:\Bn\to\pmB$ are $\epsilon$-close 
if $\pr_x[f(x) \neq g(x)] \leq \epsilon$. 

\begin{Lem}[\cite{GOS+11}]\label{lem:close to sparse}
If $f:\Bn\to\pmB$ is $\mu$-close to $s$-sparse in $\ell_2$, 
where $\mu \leq \frac{1}{20s^2}$, then $f$ is $\mu^2/2$-close to a Boolean function 
$g:\Bn\to\pmB$ of Fourier sparsity $s$. 
\end{Lem}

Putting these results together, we can prove Theorem~\ref{thm:lighttail}.
\begin{Thm-lighttail}[Restated]
If $f:\Bn\to\pmB$ is $\mu$-close to $s$-sparse in $\ell_2$, 
where $\mu \leq \frac{\log^{O(1)}\fn{f}{0}}{\sqrt{\fn{f}{0}}}$ 
and $s\leq \log^{O(1)}\fn{f}{0}$, 
then $\D_\oplus(f) \leq \log^{O(1)}\fn{f}{0}$.
\end{Thm-lighttail}
\begin{proof} 
Since $f$ is $\mu$-close to $s$-sparse, 
and $20s^2 \mu = \frac{\log^{O(1)}\fn{f}{0}}{\sqrt{\fn{f}{0}}} < 1$ for sufficiently large $\fn{f}{0}$, 
therefore by Lemma \ref{lem:close to sparse}, 
$f$ is $\mu^2/2$-close to a Boolean function $g$ of Fourier sparsity $s$. 
We will compute $f$ by computing $g$ and $fg$. 
By the setting of parameter $\mu$, it holds that 
$\mu^2/2 \leq \frac{\log^{O(1)}\fn{f}{0}}{\fn{f}{0}}$,
and hence
$\rho_{-1}(fg) \leq \frac{\log^{O(1)}\fn{f}{0}}{\fn{f}{0}}$. 
Note that $\rho_{-1}$ in the $\pmB$-range representation is 
just the same as $\rho_1$ in the $\B$-range representation. 
Applying Lemma \ref{lem:norm of product}, 
we have $\wfn{fg}{0} \leq \fn{f}{0}\fn{g}{0} \leq \fn{f}{0}\cdot s$. 
Now by Lemma \ref{lem:sparse}, we see that 
the Boolean function $fg$ can be computed using $\log^{O(1)}\fn{f}{0}$ queries. 
To compute $g$ itself, we can just use the trivial upper bound of 
$\D_\oplus (f) \leq \fn{g}{0} = s = \log^{O(1)}\fn{f}{0}$. 
Thus 
\[
\D_\oplus(f) \leq \D_\oplus(g) + \D_\oplus(fg) \leq \log^{O(1)}\fn{f}{0},
\]
as desired.
\end{proof}

\section{Concluding remarks}
The major open question is to prove Conjecture \ref{conj:rank}, $\rank(f) = O(\log^c(\fn{f}{0}))$, or even the stronger Conjecture \ref{conj:cmin by l1}, $\C_{\oplus,\min}(f) = O(\log^c\fn{f}{1}+1)$. 
In general, the gap between $\|\hat{f}\|_0$ and $\|\hat{f}\|_1$ can be huge. 
For instance, the AND function of $n$ variables have $\fn{f}{0} = 2^n$ and 
$\fn{f}{1} = O(1)$. Thus one may think that Conjecture \ref{conj:cmin by l1} is probably too strong to hold. 
However, note that the AND function has a large \gf-degree, 
and Fourier sparse functions always have \gf-degree smaller than $\log\fn{f}{0}$. 
We actually do not know any counterexample for Conjecture \ref{conj:cmin by l1} even for $c=1$. 
Indeed, we can show that Conjecture \ref{conj:cmin by l1} actually holds for several classes of functions, where for symmetric functions we use a result from \cite{AFH12}.
\begin{Prop}
	Conjecture \ref{conj:cmin by l1} is true with $c = 1$ for affine subspace indicators $\mathbbm{1}_{H}$, 
	degree-$d$ bent functions $x_1...x_d + \cdots + x_{n-d+1}...x_n$ and all symmetric functions.
\end{Prop}
Note that if $\C_{\oplus,\min}(f) = O(\log\fn{f}{1})$ is true, then we not only have $\dcc(f\circ \oplus) = O(\log^2\rank(M_{f\circ \oplus}))$, but also further improve Green-Sander's result to $T = \fn{f}{1}^d$; see Proposition \ref{prop:rankL1} and Corollary \ref{cor:dGS}. 

\medskip
In the upper bound in Lemma \ref{lem:rank by 1norm full}, 
the $2^{d^2/2}$ factor comes from the fact that $\wfn{\Delta_t f}{1} \leq \fn{f}{1}^2$. 
As we have the freedom of picking any $t$, 
is it possible that, for any Boolean function $f$, 
one can always find a $t$ such that the Fourier sparsity of its derivative $\wfn{\Delta_t f}{1}$ 
is much smaller than the trivial upper bound $\fn{f}{1}^2$? 


\newcommand{\etalchar}[1]{$^{#1}$}

\appendix
\section{Restriction of functions on affine subspace: Proof of Lemma \ref{lem:rotation}}\label{sec:restriction}
\begin{proof}(of Lemma \ref{lem:rotation})
We will show the conclusion for affine subspaces of co-dimension 1, 
and the second and third conclusion for the general $H$ follows by repeatedly applying the result. 
When $\codim(H) = 1$, namely $\codim(V) = 1$, there is a unique non-zero vector $t\in \Bn$ orthogonal to all vectors in $V$. 
Take a  basis $r^1, ..., r^{n-1}$ of $V$, 
and further take a vector $r^n\in \overline{V}$. 
Define an $n\times n$ matrix $R = [r^1, ..., r^n]$, 
then $Rf(y) = f(Ry) = f(y_1r^1 + \cdots + y_nr^n)$. 
Define two functions $f_0,f_1:\B^{n-1}\to\mbR$ by $f_b(y) = Rf(yb)$, namely 
\begin{equation}\label{eq:rotation}
f_0(y_1...y_{n-1}) = f(y_1r^1 + \cdots + y_{n-1}r^{n-1})
\quad \text{ and } \quad 
f_1(y_1...y_{n-1}) = f(y_1r^1 + \cdots + y_{n-1}r^{n-1} + r^n).
\end{equation} 
Since $f_b$ is defined on $\B^n$, 
its Fourier spectrum can be defined as before, 
and we will use it as the Fourier spectrum of $f|_H$. 
Now we will prove that this choice of definition satisfies the three conditions. 
Note that though the choice of $r^1, ..., r^n$ is not unique, 
but the vector of Fourier coefficients are the same up to a permutation, 
and in particular, its $\ell_p$-norm does not depend on the choice of $R$. 
\begin{enumerate}
\item Let us first compute the Fourier coefficients of the function $Rf:\Bn\to\mbR$. 
  It is not hard to see that for any $s\in \B^{n-1}$, we have 
\[
\widehat{Rf}(s0) = \frac{1}{2}(\hat f_0(s) + \hat f_1(s)) 
\quad \text{ and } \quad 
\widehat{Rf}(s1) = \frac{1}{2}(\hat f_0(s) - \hat f_1(s)).
\] 
This implies that 
\[
\hat f_0(s) = \hat f((R^T)^{-1}(s0)) + \hat f((R^T)^{-1}(s1)) 
\quad \text{ and } \quad
\hat f_1(s) = \hat f((R^T)^{-1}(s0)) - \hat f((R^T)^{-1}(s1)),
\] 
where we used the fact that $\widehat{Rf}(s) = \hat f((R^T)^{-1}s)$ for any invertible linear transformation $R$.
So in either subfunction, the pair of Fourier coefficients of $f$ that collide are 
\[\{(\hat f((R^T)^{-1}(s0)), \hat f((R^T)^{-1}(s1))): s\in \B^{n-1}\}.\] 
To see the relation of these two characters, suppose that the rows of 
$L=R^{-1}$ are $l^1, ..., l^n$, then by 
$LR = I$, we know that 
$\langle l^n, r^1 \rangle = ... = \langle l^n, r^{n-1} \rangle = 0$. 
Since there is only one nonzero vector, $t$, orthogonal to all $r^1$, ..., $r^{n-1}$, 
therefore $l^n = t$, and thus $(R^T)^{-1}(s0) + (R^T)^{-1}(s1) = t$. 
So the pairs are just those $(s,s+t)$. 
	
\item Since the Fourier spectrum of $f_b$ is formed by pairing up (using plus or minus) 
the Fourier spectrum of $f$, 
by the standard fact that $|a|^p + |b|^p \geq |a+b|^p$ for any $p\in [0,1]$, 
we know that $\wfn{f_b}{p} \leq \fn{f}{p}$.
	
\item First, for Boolean function $f_b:\B^{n-1}\to\pmB$, we know that 
$f_b(y) = \pm\chi_s(y) \Leftrightarrow \wfn{f_b}{0} = 1$ by definition of $\ell_0$-norm. 
Since $1 = \wfn{f_b}{2} \leq \wfn{f_b}{1} \leq \wfn{f_b}{0}$ by Parseval's Identity, 
it is easily seen that $\wfn{f_b}{1} = 1$ is equivalent to that there is only one nonzero Fourier coefficient. 
	
The conclusion now follows by noting that $f_b$ is a linear function if and only if 
$f|_H(x) = R^{-1}f_b(x) = f_b(R^{-1}x)$ is a linear function.
\end{enumerate}
\end{proof}



\newcommand{\mspan}{\mathbf{Span}}
\section{A proof of Lemma~\ref{lem:sparse}}\label{sec:BB_proof}
For any $0< \eta \leq 1$ we can define a linear operator 
$T_{\eta}:\C^{\Bn} \to \C^{\Bn}$ such that, for any 
$f:\Bn \to \C$ with Fourier expansion 
$f(x)=\sum_{t \in \Bn}\hat{f}(t)\chi_{t}(x)$,
$T_{\eta}(f)$ is a complex-valued function over the Boolean cube 
such that for every $x\in \Bn$,
\[
T_{\eta}(f)(x)=\sum_{t \in \Bn}\hat{f}(t)\eta^{|t|}\chi_{t}(x).
\]

The following remarkable theorem~\cite{Bon70, Bec75} shows that 
$T_{\eta}$ is a norm-$1$ operator from $L^{1+\eta^{2}}(\Bn)$ to $L^{2}(\Bn)$.

\begin{Thm}[Bonami-Beckner Theorem]\label{thm:BB}
Let $f:\Bn \to \C$ be a function defined over the Boolean cube. Then for every $0< \eta \leq 1$,
\[
\lVert T_{\eta}f \rVert_{2}\leq \lVert f \rVert_{1+\eta^{2}}.
\]
\end{Thm}

\begin{Lem}\label{lemma:main}
Let $f:\Bn \to \{0,1\}$ be a Boolean function with Fourier sparsity $s$ and 
density $\rho_{1}(f)=O(\polylog{s})/s$.
Then there exists an invertible linear map $L:\Bn \to \Bn$
such that all the non-zero Fourier coefficients of 
$Lf$ lie in a subspace of dimension $d=O(\polylog{s})$.
\end{Lem}

\begin{proof}
Let $\Lambda \subset \Bn$ be the set of vectors in $\Bn$ at which the 
Fourier coefficients of $f$ are non-zero. Suppose $d=\mathrm{dim}(\mspan(\Lambda))$.
Let $\xi_{1},\ldots \xi_{d}$ be a set of $d$ linearly independent vectors in $\Lambda$.
Let $L:\Bn \to \Bn$ be an invertible linear map. 
If we define a new Boolean function $f':\Bn \to \{0,1\}$
such that $f':= Lf$, then it can be readily verified that 
$\hat{f}'(\alpha)=\hat{f}((L^{-1})^{T}\alpha)$ for every $\alpha \in \Bn$. 
Therefore by choosing $L$ appropriately and replacing $f$ with $f'$ we may assume that 
$\xi_{1},\ldots,\xi_{d}$ are the standard basis $e_{1}, \ldots, e_{d}$.  
Recall that for every $\alpha \in \Lambda$, $|\hat{f}'(\alpha)| \geq 1/s$, 
then for any $0 < \eta \leq 1$, we have
\begin{alignat*}{2}
d(\frac{\eta}{s})^{2} 
&\leq \sum_{i=1}^{d}\lvert \eta \hat{f}'(e_{i})\rvert^{2} = 
      \sum_{i=1}^{d}\eta^{2}\lvert  \hat{f}'(e_{i})\rvert^{2}\\
&\leq \sum_{t \in \Lambda}  \eta^{2|t|} \lvert \hat{f}'(t)\rvert^{2} \\
&= \sum_{t \in \Lambda}\lvert \widehat{T_{\eta}f'}(t)\rvert^{2} \\
&\leq \sum_{t}\lvert \widehat{T_{\eta}f'}(t)\rvert^{2} =
     \lVert T_{\eta}f' \rVert_{2}^{2} && \text{(Parseval's Identity)}\\
&\leq \lVert f' \rVert_{1+\eta^{2}}^{2} && \text{(Theorem~\ref{thm:BB})}\\
&=\rho_{1}(f')^{\frac{2}{1+\eta^{2}}},
\end{alignat*} 
where in the last step we make use of the fact that $f'$ is a Boolean function. 
Now taking $\eta=\sqrt{\frac{\mathrm{loglog}s}{\log{s}}}$ 
gives $d\leq O(\polylog{s})$.
\end{proof}


\begin{thebibliography}{}

\end{thebibliography}


\begin{thebibliography}{ASTS{\etalchar{+}}03}

\bibitem[AFH12]{AFH12}
Anil Ada, Omar Fawzi, and Hamed Hatami.
\newblock Spectral norm of symmetric functions.
\newblock In {\em Proceedings of the 16th International Workshop on
  Randomization and Computation}, pages 338--349, 2012.

\bibitem[AKK{\etalchar{+}}05]{AKKLR05}
Noga Alon, Tali Kaufman, Michael Krivelevich, Simon Litsyn, and Dana Ron.
\newblock Testing {R}eed {M}uller codes.
\newblock {\em {IEEE} Transactions on Information Theory}, 51(11):4032--4039,
  2005.

\bibitem[ASTS{\etalchar{+}}03]{ASTS+03}
Andris Ambainis, Leonard Schulman, Amnon Ta-Shma, Umesh Vazirani, and Avi
  Wigderson.
\newblock The quantum communication complexity of sampling.
\newblock {\em SIAM Journal on Computing}, 32(6):1570--1585, 2003.

\bibitem[BC99]{BC99}
Anna Bernasconi and Bruno Codenotti.
\newblock Spectral analysis of boolean functions as a graph eigenvalue problem.
\newblock {\em IEEE Transactions on Computers}, 48(3):345--351, 1999.

\bibitem[BdW02]{BdW02}
Harry Buhrman and Ronald de~Wolf.
\newblock Complexity measures and decision tree complexity: a survey.
\newblock {\em Theoretical Computer Science}, 288(1):21--43, 2002.

\bibitem[Bec75]{Bec75}
William Beckner.
\newblock Inequalities in {Fourier} analysis.
\newblock {\em Annals of Mathematics}, 102:159--182, 1975.

\bibitem[Bon70]{Bon70}
Aline Bonami.
\newblock {\'{E}}tude des coefficients {F}ourier des fonctiones de
  ${L}^{p}({G})$.
\newblock {\em Annales de l'institut {F}ourier}, 20(2):335--402, 1970.

\bibitem[BSLRZ12]{BLR12}
Eli Ben-Sasson, Shachar Lovett, and Noga Ron-Zewi.
\newblock An additive combinatorics approach relating rank to communication
  complexity.
\newblock In {\em Proceedings of The 53rd Annual IEEE Symposium on Foundations
  of Computer Science}, pages 177--186, 2012.

\bibitem[Cha02]{Cha02}
Mei-Chu Chang.
\newblock A polynomial bound in {F}reiman's theorem.
\newblock {\em Duke Mathematical Journal}, 113(3):399--419, 2002.

\bibitem[Dic58]{Dic58}
L.~Dickson.
\newblock {\em Linear groups with an exposition of the Galois field theory}.
\newblock Dover, New York, 1958.

\bibitem[GOS{\etalchar{+}}11]{GOS+11}
Parikshit Gopalan, Ryan O'Donnell, Rocco Servedio, Amir Shpilka, and Karl
  Wimme.
\newblock Testing {F}ourier dimensionality and sparsity.
\newblock {\em SIAM Journal on Computing}, 40(4):1075--1100, 2011.

\bibitem[Gow98]{Gow98}
Timothy Gowers.
\newblock A new proof of {S}zemer\'edi's theorem for arithmetic progressions of
  length four.
\newblock {\em Geometric and Functional Analysis}, 8(3):529--551, 1998.

\bibitem[Gow01]{Gow01}
Timothy Gowers.
\newblock A new proof of {S}zemer\'edi's theorem.
\newblock {\em Geometric and Functional Analysis}, 11(3):465--588, 2001.

\bibitem[Gro97]{Gro97}
Vince Grolmusz.
\newblock On the power of circuits with gates of low {L1} norms.
\newblock {\em Theoretical Computer Science}, 188(1-2):117--128, 1997.

\bibitem[GS08]{GS08}
Ben Green and Tom Sanders.
\newblock Boolean functions with small spectral norm.
\newblock {\em Geometric and Functional Analysis}, 18(1):144--162, 2008.

\bibitem[GT09]{GT09}
Ben Green and Terence Tao.
\newblock The distribution of polynomials over finite fields, with applications
  to the gowers norms.
\newblock {\em Contributions to Discrete Mathematics}, 4(2):1--36, 2009.

\bibitem[HS10]{HS10}
Elad Haramaty and Amir Shpilka.
\newblock On the structure of cubic and quartic polynomials.
\newblock In {\em Proceedings of the 42nd ACM Symposium on Theory of
  Computing}, pages 331--340, 2010.

\bibitem[IMR12]{IMR12}
Russell Impagliazzo, Cristopher Moore, and Alexander Russell.
\newblock An entropic proof of {C}hang's inequality.
\newblock {\em arXiv:1205.0263}, 2012.

\bibitem[KL96]{KL96}
Andrew Kotlov and L{\'a}szl{\'o} Lov{\'a}sz.
\newblock The rank and size of graphs.
\newblock {\em Journal of Graph Theory}, 23:185--189, 1996.

\bibitem[KL08]{KL08}
Tali Kaufman and Shachar Lovett.
\newblock Worst case to average case reductions for polynomials.
\newblock In {\em Proceedings of the 49th Annual IEEE Symposium on Foundations
  of Computer Science}, pages 166--175, 2008.

\bibitem[KM93]{KM93}
Eyal Kushilevitz and Yishay Mansour.
\newblock Learning decision trees using the f{}ourier spectrum.
\newblock {\em SIAM Journal on Computing}, 22(6):1331--1348, 1993.

\bibitem[KN97]{KN97}
Eyal Kushilevitz and Noam Nisan.
\newblock {\em Communication Complexity}.
\newblock Cambridge University Press, Cambridge, UK, 1997.

\bibitem[Kot97]{Kot97}
Andrei Kotlov.
\newblock Rank and chromatic number of a graph.
\newblock {\em Journal of Graph Theory}, 26:1--8, 1997.

\bibitem[KS13]{KS13}
Raghav Kulkarni and Miklos Santha.
\newblock Query complexity of matroids.
\newblock In {\em Proceedings of the 8th International Conference on Algorithms
  and Complexity}, 2013.

\bibitem[LLZ11]{LLZ11}
Ming~Lam Leung, Yang Li, and Shengyu Zhang.
\newblock Tight bounds on the communication complexity of symmetric {XOR}
  functions in one-way and {SMP} models.
\newblock In {\em Proceedings of the 8th Annual Conference on Theory and
  Applications of Models of Computation}, pages 403--408, 2011.

\bibitem[Lov90]{Lov90}
L{\'a}szl{\'o} Lov{\'a}sz.
\newblock Communication complexity: A survey.
\newblock In {\em B. Korte, L. Lov{\'a}sz, H. Prömel, and A. Schrijver,
  editors, Paths, flows, and VLSI-layout, Springer-Verlag}, page 235–265,
  1990.

\bibitem[LS88]{LS88}
L{\'a}szl{\'o} Lov{\'a}sz and Michael~E. Saks.
\newblock Lattices, {M}{\"o}bius functions and communication complexity.
\newblock In {\em Proceedings of the 29th Annual Symposium on Foundations of
  Computer Science}, pages 81--90, 1988.

\bibitem[LS09]{LS09}
Troy Lee and Adi Shraibman.
\newblock Lower bounds on communication complexity.
\newblock {\em Foundations and Trends in Theoretical Computer Science},
  3(4):263--398, 2009.

\bibitem[LZ10]{LZ10}
Troy Lee and Shengyu Zhang.
\newblock Composition theorems in communication complexity.
\newblock In {\em Proceedings of the 37th International Colloquium on Automata,
  Languages and Programming (ICALP)}, pages 475--489, 2010.

\bibitem[LZ13]{LZ13}
Yang Liu and Shengyu Zhang.
\newblock Quantum and randomized communication complexity of {XOR} functions in
  the {SMP} model.
\newblock {\em ECCC}, 20(10), 2013.

\bibitem[MO10]{MO10}
Ashley Montanaro and Tobias Osborne.
\newblock On the communication complexity of {XOR} functions.
\newblock {\em arXiv:}, 0909.3392v2, 2010.

\bibitem[MS82]{MS82}
Kurt Mehlhorn and Erik~M. Schmidt.
\newblock Las {V}egas is better than determinism in {VLSI} and distributed
  computing (extended abstract).
\newblock In {\em Proceedings of the fourteenth annual ACM symposium on Theory
  of computing}, pages 330--337, 1982.

\bibitem[NW95]{NW95}
Noam Nisan and Avi Wigderson.
\newblock On rank vs. communication complexity.
\newblock {\em Combinatorica}, 15(4):557--565, 1995.

\bibitem[O'D12]{O12}
Ryan O'Donnell.
\newblock Lecture notes: Analysis of {B}oolean functions. {15-859S}, {Carnegie
  Mellon University}, 2012.

\bibitem[SV13]{SV13}
Amir Shpilka and Ben~Lee Volk.
\newblock On the structure of boolean functions with small spectral norm.
\newblock {\em ECCC}, TR13-049, 2013.

\bibitem[SW12]{SW12}
Xiaoming Sun and Chengu Wang.
\newblock Randomized communication complexity for linear algebra problems over
  finite fields.
\newblock In {\em Proceedings of the 29th International Symposium on
  Theoretical Aspects of Computer Science}, pages 477--488, 2012.

\bibitem[Val04]{Val04}
Paul Valiant.
\newblock The log-rank conjecture and low degree polynomials.
\newblock {\em Information Processing Letters}, 89(2):99--103, 2004.

\bibitem[Yao79]{Yao79}
Andrew Yao.
\newblock Some complexity questions related to distributive computing.
\newblock In {\em Proceedings of the Eleventh Annual ACM Symposium on Theory of
  Computing (STOC)}, pages 209--213, 1979.

\bibitem[ZS09]{ZS09}
Zhiqiang Zhang and Yaoyun Shi.
\newblock Communication complexities of symmetric {XOR} functions.
\newblock {\em Quantum Information {\&} Computation}, 9(3):255--263, 2009.

\bibitem[ZS10]{ZS10}
Zhiqiang Zhang and Yaoyun Shi.
\newblock On the parity complexity measures of {B}oolean functions.
\newblock {\em Theoretical Computer Science}, 411(26-28):2612--2618, 2010.

\end{thebibliography}
\end{document}